\documentclass[aps,pra,twocolumn,superscriptaddress,nofootinbib]{revtex4-2}

\usepackage[utf8]{inputenc}
\usepackage[T1]{fontenc}
\usepackage{lmodern}

\usepackage{amsmath,amssymb,amsfonts,amsthm}
\usepackage{mathtools}
\usepackage{physics}
\usepackage{bm}
\usepackage{bbm}
\usepackage{graphicx}
\usepackage{xcolor}
\usepackage{subcaption}
\usepackage{hyperref}
\hypersetup{
  colorlinks=true,
  linkcolor=blue,
  citecolor=blue,
  urlcolor=blue
}

\usepackage{microtype}
\usepackage{orcidlink}
\usepackage[colorinlistoftodos]{todonotes}

\newtheorem{theorem}{Theorem}
\newtheorem{proposition}{Proposition}

\theoremstyle{definition}
\newtheorem{definition}{Definition}
\newtheorem{property}{Property}
\theoremstyle{remark}

\newcommand{\cU}{\mathcal{U}}

\newcommand{\Favg}{\overline{F}}             
\newcommand{\Fprod}{\overline{F}^{\otimes}}  

\newcommand{\id}{\mathbbm{1}}
\newcommand{\C}{\mathbb{C}}
\newcommand{\Ad}{\mathrm{Ad}}

\newcommand{\rd}[1]{{\color{black}#1}}

\begin{document}

\title{Entangling power and fidelity diagnostics for bipartite quantum channels}

\author{Marcin Rudzi\'nski \orcidlink{0000-0002-6638-3978}}
\affiliation{Faculty of Physics, Astronomy and Applied Computer Science, Jagiellonian University, ul. \L{}ojasiewicza 11, 30-348 Krak{\'o}w, Poland}
\affiliation{Doctoral School of Exact and Natural Sciences, Jagiellonian University, ul. \L{}ojasiewicza 11, 30-348 Krak{\'o}w, Poland}

\author{Gianluigi Tartaglione
\orcidlink{0009-0003-1594-4733}}

\affiliation{Department of Industrial Engineering, University of Salerno, via Giovanni Paolo II 132, 84084 Fisciano, Salerno, Italy}
\affiliation{Institute of Nanotechnology of the National Research Council of Italy, CNR-NANOTEC, Lecce Central Unit, c/o Campus Ecotekne, Via Monteroni, 73100 Lecce, Italy}

\author{Karol {\.Z}yczkowski \orcidlink{0000-0002-0653-3639}}
\affiliation{Faculty of Physics, Astronomy and Applied Computer Science, Jagiellonian University, ul. \L{}ojasiewicza 11, 30-348 Krak{\'o}w, Poland}
\affiliation{Center for Theoretical Physics, Polish Academy of Sciences,  Al. Lotnik{\'o}w 32/46, 02-668 Warszawa, Poland}

\date{\rd{August 7, 2026}}

\begin{abstract}
\rd{We study two distinct diagnostics of bipartite quantum channels: the average fidelity between a desired state and noisy output across different classes of input states, and generation of entanglement from pure product input states, quantified by well-defined entangling power. These diagnostics are generally independent, but together they provide a more complete description of bipartite channels.} We show that the fidelity averaged over any local-unitary orbit with fixed Schmidt coefficients for equal local dimensions is completely determined by the average input-output fidelity and its restriction to product inputs.
We also introduce concurrence- and negativity-based entangling power for two-qubit channels, prove their convexity and monotonicity under local postprocessing, and show that, unlike the previously proposed linear-entropy quantity, they vanish for all separable channels. 
Representative separable and nonseparable channels are investigated.
Finally, we generalize our definitions of entangling power to nonproduct inputs, and provide an analytical lower bound for the concurrence-based entanglement variation, showing its effectiveness with specific examples. 
\end{abstract}

\maketitle
\section{Introduction}
\label{sec:intro}

Entanglement is a central nonclassical feature of composite quantum systems and a key resource in quantum information processing \cite{Horodecki2009}. In gate-based architectures, nonlocal processing is built from one- and two-qubit gates.  Two-qubit gates are already universal when combined with arbitrary single-qubit operations \cite{DiVincenzo1995,Barenco1995,Bremner2002}. The controlled generation and preservation of bipartite entanglement are therefore essential for scalable quantum computation. However, in realistic devices native operations are noisy and therefore described by quantum channels \cite{NielsenChuang}. This issue is especially relevant in the noisy
intermediate-scale quantum (NISQ) era, where gate errors and decoherence strongly constrain circuit depth and performance \cite{Preskill2018}. 
\rd{
This motivates two separate questions: how well a noisy bipartite operation implements a desired transformation, and how much entanglement 
it generates on average for initially pure random product states.
We study both issues because they provide nonredundant information about a noisy bipartite gate.}

Conventional quantifiers show how close a noisy process is to a target operation rather than how strongly it creates entanglement from initially uncorrelated inputs. Average gate fidelity, randomized benchmarking, and unitarity are now standard probes of control quality and noise coherence \cite{Nielsen2002,Bowdrey2002,Magesan2012,Wallman2015}.
The more general problem of studying the nonlocal properties of bipartite unitary channels dates back to the work of Makhlin \cite{Makhlin2002}, who provided a set of invariants characterizing the nonlocality of unitary operations.
The optimal creation of entanglement by a two-qubit unitary operation, including ancilla-assisted settings, was studied in Ref.~\cite{Kraus2001}.

Entangling power was introduced by Zanardi, Zalka, and Faoro \cite{Zanardi2000} as the average entanglement generated by a bipartite unitary acting on pure product states. This quantity was later connected to operator entanglement and was used to extend the framework to higher-dimensional and ancilla-assisted settings \cite{Zanardi2001, Wang2003,Jonnadula2017}. It is also known that allowing initially entangled inputs can enhance the entanglement-generating capacity of a quantum operation \cite{Leifer2003}. For nonunitary dynamics, a natural extension based on the linear entropy of the reduced output state was also proposed \cite{Kong2015}. Although this definition reduces to the standard pure-state formula in the unitary case, for mixed outputs, linear entropy measures local impurity rather than entanglement itself. As a result, it can assign a positive value to channels which do not generate entanglement from product inputs. 
This limits its use as a diagnostic of genuine entanglement generation.

Furthermore, the nonlocal properties of bipartite nonunitary quantum channels have recently become a subject of increasing interest. For instance, the capability of bipartite channels to generate maximally entangled states and secret keys has been investigated \cite{Bauml2018}. Resource-theoretic approaches have also been developed for the entanglement of channels \cite{Gour2021} and for nonlocality under local operations and shared randomness (LOSR) \cite{Schmid2020}, the latter providing a unifying framework encompassing different types of nonlocal resources.
Certification of the nonlocality of bipartite channels from the induced probability distributions
\cite{Rico2025}, as well as the quantification of type-independent nonlocality under the effect of noise \cite{Rico2026} and moment-based characterizations using
local-unitary invariants~\cite{ShaglelImai2026}, have also been studied.

The main goal of this work is to return to the notion of entangling power for nonunitary operations and to separate genuine mixed-state entanglement generation from local impurity effects. The linear-entropy extension of Ref.~\cite{Kong2015} is natural by analogy with the unitary pure-output case, but for mixed outputs it can be positive even when the channel does not entangle product inputs.
We therefore replace the linear-entropy output average by quantities built from genuine mixed-state entanglement monotones.
 In the two-qubit case, concurrence and negativity are especially natural choices. Both are standard computable entanglement monotones, and both vanish on separable states \cite{Hill1997,Wootters1998,VidalWerner2002,Horodecki2009}. Concurrence is closely tied to two-qubit entanglement through the convex-roof construction \cite{Wootters1998}, while negativity \cite{Zyczkowski1998,Zyczkowski1999} is directly computable from the partial transpose and is closely related to the Peres-Horodecki criterion \cite{Peres1996,Horodecki1996,VidalWerner2002}. For two qubits, the positive partial transpose (PPT) criterion is necessary and sufficient for separability, so negativity detects exactly the entangled outputs. In higher dimensions, negativity remains a useful negative partial transpose (NPT)-entanglement witness and monotone, but it does not detect PPT-bound entangled states.

At the same time, fidelity-based diagnostics remain informative. Besides the usual average input-output fidelity over all pure states, one may consider the corresponding average restricted to product inputs. In this paper we show that these two fidelities determine the full fixed Schmidt vector local-unitary orbit average fidelity for equal local dimensions. 
This makes it possible to distinguish channels that favor lower-purity input orbits from those that favor product inputs, even for channels with zero proper entangling power. For two qubits, reduced purity fixes the Schmidt spectrum up to permutation, so this is equivalently a distinction between more entangled and more productlike pure inputs. Fidelity diagnostics and entanglement generation therefore capture \rd{distinct, nonredundant} aspects of bipartite noisy dynamics.

The paper is organized as follows.  In Sec.~\ref{sec:mean-fidelity} we introduce the average and product-state input-output fidelities, derive closed Kraus-operator formulas for them, and show that, for equal local dimensions, they determine the Schmidt-orbit averaged fidelity profile. In Sec.~\ref{sec:entangling-power} we define concurrence- and negativity-based entangling powers, compare them with the linear-entropy-based extension of Ref.~\cite{Kong2015}, and establish their basic structural properties. Next in Sec.~\ref{sec:case-studies} we analyze several exactly solvable noise models and controlled-phase benchmarks. In Sec.~\ref{sec:bounds} we derive analytic bounds based on the tangle, and finally in Sec.~\ref{sec:variation} we extend the framework to average entanglement variation over Schmidt orbits.

\section{Input-output fidelity}\label{sec:mean-fidelity}
For a noisy implementation of the identity channel, a natural first question is how well the implemented channel preserves input states on average. For a noisy implementation of a target unitary, the analogous question concerns the corresponding error channel.
For a bipartite channel, there are at least two physically distinct notions of such an average. One may average over all pure input states, or restrict the average to product inputs only. The difference between these two quantities is precisely what later allows us to diagnose whether a given channel treats different Schmidt classes of input states in different ways.
\rd{This comparison concerns input-dependent fidelity and does not by itself diagnose entanglement generation.}

Let $\Phi:\mathcal B(\C^{d_A}\otimes\C^{d_B})\to \mathcal B(\C^{d_A}\otimes\C^{d_B})$ be a completely positive and trace-preserving (CPTP) map with Kraus representation
\begin{equation*}
\Phi(\rho)=\sum_\alpha K_\alpha \rho K_\alpha^\dagger,
\qquad
\sum_\alpha K_\alpha^\dagger K_\alpha=\id.
\end{equation*}
We recall the average input-output fidelity \cite{Nielsen2002,Bowdrey2002},
\begin{equation}
\Favg(\Phi)
:=
\int d\psi\;
\bra{\psi}\Phi(\ketbra{\psi}{\psi})\ket{\psi},
\label{eq:def-Favg}
\end{equation}
where the integral is taken over Haar-random pure states on $\C^{d_A}\otimes\C^{d_B}$. This is the standard average gate fidelity of $\Phi$ with respect to the identity channel \cite{Nielsen2002,Bowdrey2002}. 
For a noisy implementation of a target unitary channel $\cU=\Ad_U$, where
$\Ad_U(\rho):=U\rho U^\dagger$,
one simply applies the same definition to the error channel $\Ad_{U^\dagger}\circ\Phi$,
\begin{equation}
\Favg(\Phi,\cU)
:=
\Favg(\Ad_{U^\dagger}\circ\Phi),
\label{eq:def-Favg-relative}
\end{equation}
and similarly for all other fidelity quantities introduced below.
The average product-state input-output fidelity reads
\begin{align}
\Fprod(\Phi)\hspace{-1mm}
:=\hspace{-1mm}\int \hspace{-1mm} d\psi_A\,d\psi_B
\bra{\psi_A\psi_B}\Phi(\ketbra{\psi_A\psi_B}{\psi_A\psi_B})\ket{\psi_A\psi_B},
\label{eq:def-Fprod}
\end{align}
where the integration is over independent Haar-random pure states on $\C^{d_A}$ and $\C^{d_B}$. The quantity $\Favg$ probes average state preservation over the full Hilbert space, whereas $\Fprod$ describes the sector that is operationally most accessible in many experiments, and provides a direct comparison point for entangled inputs.
Both averages admit simple closed forms in terms of Kraus operators.

\begin{proposition}
\label{prop:Favg-Fprod}
For any bipartite channel $\Phi$,
\begin{equation}
\Favg(\Phi)=\frac{D+T}{D(D+1)}, 
\label{eq:Favg-kraus}
\end{equation}
where \[D=d_A d_B\quad \text{and}\quad T:=\sum_\alpha |\Tr K_\alpha|^2.\]
Moreover,
\begin{equation}
\Fprod(\Phi)=\frac{D+T+M_A+M_B}{D(d_A+1)(d_B+1)},
\label{eq:Fprod-kraus}
\end{equation}
where reduced Kraus operators read
\begin{equation*}
K_\alpha^A:=\Tr_B K_\alpha,\qquad
K_\alpha^B:=\Tr_A K_\alpha,
\end{equation*}
and the following symbols are used:
\begin{equation*}
M_A:=\sum_\alpha \Tr(K_\alpha^A K_\alpha^{A\dagger}),\quad
M_B:=\sum_\alpha \Tr(K_\alpha^B K_\alpha^{B\dagger}).
\end{equation*}
\end{proposition}

\begin{proof}
Equation~\eqref{eq:Favg-kraus} is the standard average-gate-fidelity formula \cite{Nielsen2002,Bowdrey2002}\rd{, while Eq.~\eqref{eq:Fprod-kraus} is equivalent to the bipartite specialization of the more general formula for multipartite channels \cite{Demianowicz2010}}. For completeness, both Eqs.~\eqref{eq:Favg-kraus} and \eqref{eq:Fprod-kraus} follow from the second-moment Haar identity, which can also be derived by Weingarten calculus \cite{CollinsSniady2006}. 
For further details see Appendix~\ref{app:fidelity-derivations}.
\end{proof}

The structure of Eqs. \eqref{eq:Favg-kraus} and \eqref{eq:Fprod-kraus} is already suggestive. The quantity $T$ is the global part that also appears in the ordinary average gate fidelity, while $M_A$ and $M_B$ arise only when the average is restricted to product states. In this sense, $\Fprod$ resolves additional one-body structure that is invisible in $\Favg$ alone.

To probe the dependence on the entanglement of the input state, we use Schmidt orbits.  
Throughout this part of the section we consider equal local dimensions, i.e.,
\(d_A=d_B=d\).
Let
\begin{equation}
\ket{\phi_{\vec\lambda}}
:=
\sum_{i=1}^d \sqrt{\lambda_i}\,\ket{ii},
\qquad
\lambda_i\ge 0,
\qquad
\sum_{i=1}^d \lambda_i=1,
\label{eq:schmidt_state_lambda}
\end{equation}
be a pure state in a Schmidt basis with a Schmidt vector
\(\vec\lambda=(\lambda_1,\dots,\lambda_d)\), and define its local-unitary orbit
\begin{equation}
\mathcal S_{\vec\lambda}
:=
\left\{
(U_A\otimes U_B)\ket{\phi_{\vec\lambda}}
:
U_A,U_B\in U(d)
\right\}.
\label{eq:schmidt_orbit_lambda}
\end{equation}
The product orbit corresponds to \(\lambda_1=1\), whereas the maximally entangled orbit corresponds to \(\lambda_i=1/d\) for all \(i\).  
It is convenient to parametrize the orbit by the reduced-state purity,
\begin{equation}
\mu(\vec\lambda)
:=
\Tr(\rho_A^2)
=
\sum_{i=1}^d \lambda_i^2,
\qquad
\rho_A=\Tr_B\ketbra{\phi_{\vec\lambda}}{\phi_{\vec\lambda}}.
\label{eq:def-mu-lambda}
\end{equation}
We define the orbit-resolved average fidelity by
\begin{widetext}
\begin{equation}
\Favg_{\vec\lambda}(\Phi)
:=
\int dU_A\,dU_B\;
\bra{\phi_{\vec\lambda}}
(U_A^\dagger\otimes U_B^\dagger)\,
\Phi\!\left(
(U_A\otimes U_B)\ketbra{\phi_{\vec\lambda}}{\phi_{\vec\lambda}}
(U_A^\dagger\otimes U_B^\dagger)
\right)
(U_A\otimes U_B)\ket{\phi_{\vec\lambda}}.
\label{eq:def-Flambda}
\end{equation}
\end{widetext}
The key point is that, although local-unitary orbits in \(d\times d\) are labeled by the full Schmidt vector \(\vec\lambda\), the averaged fidelity depends on \(\vec\lambda\) only through the single scalar \(\mu(\vec\lambda)\).

\begin{proposition}
\label{prop:Flambda-Kraus}
For any bipartite channel, with equal local dimension \(d\), i.e., \(\Phi:\mathcal B(\mathbb C^d\otimes\mathbb C^d)\to
\mathcal B(\mathbb C^d\otimes\mathbb C^d)\),
\begin{align}
\Favg_{\vec\lambda}(\Phi)
=
\frac{1}{(d^2-1)^2}
\Bigg[
(d^2+T)\left(1+\frac{1}{d^2}-\frac{2\mu(\vec\lambda)}{d}\right)
\nonumber
\\+(M_A+M_B)\left(\left(1+\frac{1}{d^2}\right)\mu(\vec\lambda)-\frac{2}{d}\right)
\Bigg],
\label{eq:Flambda-Kraus}
\end{align}
where \(T\), \(M_A\), and \(M_B\) are the Kraus invariants from
Proposition~\ref{prop:Favg-Fprod}.
\end{proposition}
The proof is presented in Appendix~\ref{app:fidelity-schmidt-qudit}. Eliminating \(T\), \(M_A\), and \(M_B\) in favor of \(\Favg\) and \(\Fprod\), direct calculation gives a simple interpolation formula.
\begin{theorem}
\label{thm:orbit-theorem}
For any channel \(\Phi\) on \(\mathbb C^d\otimes\mathbb C^d\) and any Schmidt vector \(\vec\lambda\), the average input-output fidelity of an orbit reads
\begin{align}\label{eq:orbit-theorem-qudit}
	\Favg_{\vec\lambda}(\Phi)
	&=
	\Fprod(\Phi)
	+
	\frac{d^2+1}{(d-1)^2}(1-\mu)\,\chi_F(\Phi),
	\\
	\chi_F(\Phi)&:=\Favg(\Phi)-\Fprod(\Phi),\nonumber
\end{align}
where the reduced-state purity \(\mu=\mu(\vec\lambda)\) depends on the Schmidt vector \(\vec{\lambda}\), and  $\chi_F$ is the fidelity-bias parameter.
\end{theorem}
Thus, the 
fidelity averaged over an orbit
is controlled by the two experimentally meaningful averages \(\Favg\) and \(\Fprod\), together with the orbit parameter \(\mu(\vec\lambda)\). 
At fixed $\Fprod$, the slope of the orbit-averaged fidelity with respect to the input linear entropy $1- \mu(\vec \lambda
)$ is $\frac{d^2+1}{(d-1)^2}\chi_F$. The parameter $\chi_F$ is neither a witness
of output-state entanglement nor an indicator of entangling power.
Since $\mu(\vec\lambda)$ is the purity of the reduced state [Eqs. \eqref{eq:def-mu-lambda}], the sign of \(\chi_F\) has the following operational meaning:
\begin{align*}
	\chi_F>0 \Rightarrow \ &\text{lower-purity input orbits are, on average,} 
		\\&\text{more robust against noise},
		\\\chi_F<0 \Rightarrow \ & \text{productlike inputs are, on average,} \\&\text{more robust against noise},\\
	\chi_F=0 \Rightarrow \ &\text{no purity-resolved orbit bias at this level}.
\end{align*}
\rd{For two qubits, the Schmidt vector has only one independent
	parameter and can be written as
	\(\vec{\lambda}=(\cos^2\theta,\sin^2\theta)\). Hence}
\begin{equation*}
	\rd{
		\mu(\vec{\lambda})
		=
		\cos^4\theta+\sin^4\theta
		=
		1-\frac{1}{2}\sin^2(2\theta),
	}
\end{equation*}
\rd{and Theorem~\ref{thm:orbit-theorem} reduces to}
\begin{equation}
	F_{\theta}(\Phi)
	=
	F^\otimes(\Phi)
	+
	\frac{5}{2}\sin^2(2\theta)
	\left(
	F(\Phi)-F^\otimes(\Phi)
	\right).
\end{equation}

\rd{In this case, \(\mu\) uniquely labels the local-unitary orbit. For \(d>2\), distinct Schmidt spectra can have the same value of
	\(\mu\) and their orbit-averaged fidelities coincide.}
	
\rd{We next address a separate question: under what conditions can
	the sign of \(\chi_F\) be constrained for independently acting local
	noise? For a broad class of product channels, such a constraint can
	be obtained from the single-qubit average fidelities.}
\begin{proposition}
	\label{prop:product-channel-chi}
	Let
	\begin{equation}
		\Phi=\Lambda_A\otimes\Lambda_B
		\label{eq:product-channel}
	\end{equation}
	be a two-qubit product channel, and let
	\begin{equation}
		\Lambda_A(\rho)=\sum_i A_i \rho A_i^\dagger,
		\qquad
		\Lambda_B(\rho)=\sum_j B_j \rho B_j^\dagger,
	\end{equation}
	be Kraus representations of the local channels. Define
	\begin{equation}
		x_A:=\sum_i |\Tr A_i|^2,
		\qquad
		x_B:=\sum_j |\Tr B_j|^2.
		\label{eq:def-xA-xB}
	\end{equation}
	Equivalently, in terms of the single-qubit average fidelities,
	\begin{equation}
		f_A:=\Favg(\Lambda_A)=\frac{2+x_A}{6},
		\
		f_B:=\Favg(\Lambda_B)=\frac{2+x_B}{6},
		\label{eq:def-fA-fB}
	\end{equation}
	For a two-qubit product channel $\Phi=\Lambda_A\otimes\Lambda_B$,
	\begin{equation}
		\chi_F(\Phi)
		=
		\frac{8+2x_Ax_B-5x_A-5x_B}{90}.
		\label{eq:chi-product-x}
	\end{equation}
	Equivalently,
	\begin{align}
		\chi_F&(\Phi)
		=
		\frac{2+4f_Af_B-3f_A-3f_B}{5}
		\nonumber\\&=
		\frac{4(1-f_A)(1-f_B)-(1-f_A)-(1-f_B)}{5}.
		\label{eq:chi-product-f}
	\end{align}
	In particular, if
	\begin{equation}
		f_A\ge \frac12,
		\qquad
		f_B\ge \frac12,
		\label{eq:condition-fA-fB}
	\end{equation}
	then
	\begin{equation}
		\chi_F(\Lambda_A\otimes\Lambda_B)\le 0.
		\label{eq:chi-product-negative}
	\end{equation}
\end{proposition}
The proof is presented in Appendix~\ref{app:fidelity-derivations}.
Assumption \eqref{eq:condition-fA-fB} is essential. 
For example, the local unitary error channel $\Ad_{\sigma_x}\otimes\Ad_{\sigma_x}$, where $\sigma_x$ is a Pauli-X gate,  has
\begin{equation}
	f_A=f_B=\frac13,
	\qquad
	\chi_F(\Ad_{\sigma_x}\otimes\Ad_{\sigma_x})=\frac{4}{45}>0.
	\label{eq:local-unitary-counterexample}
\end{equation}
\rd{Thus, \(\chi_F\) is not a witness of entanglement generation. Proposition~\ref{prop:product-channel-chi} should instead be read as a conditional test for independently calibrated product noise. If the factorization \(\Phi=\Lambda_A\otimes\Lambda_B\) is physically justified and the local fidelities are known from separate single-qubit calibration, condition~\eqref{eq:condition-fA-fB} predicts \(\chi_F\leq0\). Observing \(\chi_F>0\) would then rule out this product-noise model, but would not certify entanglement generation, as the correlated-dephasing example below demonstrates.}

\section{Entangling power for bipartite channels}\label{sec:entangling-power}

Entangling power for nonunitary evolution was defined by averaging the linear entropy of the reduced output over product pure inputs \cite{Kong2015}:
\begin{align}
e_L(\Phi)&:=\int d\psi_A d\psi_B\ E_L\!\left(\Phi(\ketbra{\psi_A\psi_B}{\psi_A\psi_B})\right),\\
E_L(\rho_{AB})&:=1-\Tr(\rho_A^2),\qquad \rho_A=\Tr_B \rho_{AB}.
\end{align}
For a pure state $\rho_{AB}$, $E_L$ is an entanglement monotone, while
for a mixed state $\rho_{AB}$, $E_L$ measures local impurity. Therefore, it can be nonzero even when $\rho_{AB}$ is separable, for instance, for the maximally mixed state $\rho_{AB}=\mathbbm{1}/(d_A d_B)$.
Hence \(e_L(\Phi)\) can be nonzero not only for channels that fail to create entanglement on average, but even for separable channels whose outputs are separable for every product input.
To deal with this issue, one can keep the same average over product inputs but replace the reduced-state linear entropy by a genuine mixed-state entanglement monotone.

\begin{definition}
	\label{def:general-entangling-power}
	\rd{Let \(M\) be a mixed-state entanglement monotone for bipartite
		states on \(\mathcal{H}_A\otimes\mathcal{H}_B\). In particular, we
		assume that \(M(\rho_{AB})\geq 0\), that \(M\) vanishes on every
		separable state, and that it does not increase under local operations
		and classical communication. The \(M\)-based entangling power of a
		bipartite channel
		is defined as}
	\begin{equation}
		\rd{\displaystyle
			e_M(\Phi):=
			\int d\psi_A\,d\psi_B\,
			M\!\left(
			\Phi\!\left(
			|\psi_A\psi_B\rangle\langle\psi_A\psi_B|
			\right)
			\right).
		}
		\label{eq:general-entangling-power}
	\end{equation}
	\rd{Here, \(d\psi_A\) and \(d\psi_B\) denote the normalized Haar
		measures on the pure states of the two local subsystems.}
\end{definition}

\rd{Definition~\ref{def:general-entangling-power} applies to arbitrary local dimensions and to any entanglement monotone satisfying the assumptions above. Since every separable channel maps a product input to a
	separable output, its entangling power is zero according to this definition. The converse need not hold, since a channel may preserve
	separability of all product inputs without itself admitting a
	separable Kraus representation.}

For two qubits, standard mixed-state entanglement monotones \cite{Plenio2007} include concurrence $C(\rho)$ and negativity $\mathcal{N}(\rho)$.
Therefore, we introduce the following quantities:
\begin{definition}
Concurrence-based entangling power reads
    \begin{equation}
e_C(\Phi):=\int d\psi_A d\psi_B\ C\!\left(\Phi(\ketbra{\psi_A\psi_B}{\psi_A\psi_B})\right),
\label{eq:eC}
\end{equation}
where concurrence is given by
\[C(\rho):=\max\{0,r_1-r_2-r_3-r_4\}.\] 
Here \(r_1\geq r_2\geq r_3\geq r_4\geq0\) are eigenvalues of the matrix \(R=\sqrt{\sqrt{\rho}\tilde{\rho}\sqrt{\rho}}\), where \(\tilde{\rho}=(\sigma_y\otimes\sigma_y)\rho^*(\sigma_y\otimes\sigma_y)\).
\end{definition}

\begin{definition}
    Negativity-based entangling power is
\begin{align}
e_{\mathcal{N}}(\Phi):=\int d\psi_A d\psi_B\ \mathcal{N}\!\left(\Phi(\ketbra{\psi_A\psi_B}{\psi_A\psi_B})\right),
\label{eq:eN}
\end{align}
where negativity is given by \cite{Zyczkowski1998}
\[\mathcal{N}(\rho):=\frac{\|\rho^{T_B}\|_1-1}{2}.\]
Here $\rho^{T_B}$ denotes the partial transpose on $B$ and $\|\cdot\|_1$ is the trace norm. 
\end{definition}

We next list several basic properties of the introduced quantities.
\begin{property}
    The introduced entangling-power measures satisfy
\begin{equation}
e_{\mathcal N}(\Phi)\le \frac12\,e_C(\Phi).
\label{eq:eN-le-half-eC}
\end{equation}
It follows from the inequality for concurrence and negativity, $2\mathcal{N}(\rho)\leq C(\rho)$  \cite{Verstraete2001},  with equality for pure states. 
\end{property}

\begin{property}
    Assuming $\Phi$ is separable, i.e., it admits a Kraus form with product Kraus operators
\begin{equation}
\Phi(\rho)=\sum_k (A_k\otimes B_k)\rho(A_k\otimes B_k)^\dagger,
\label{eq:sepchannel}
\end{equation}
then $\Phi$ maps every product pure input to a separable output, so
\begin{equation}
e_C(\Phi)=0,\qquad e_{\mathcal{N}}(\Phi)=0.
\label{eq:sep_zero}
\end{equation}
Note that the converse statement is not true in general. 
\rd{The \textsc{swap} unitary maps every product input $\ket{\psi_A}\ket{\psi_B}$ to the product state $\ket{\psi_B}\ket{\psi_A}$ and hence has zero entangling power, although the \textsc{swap} channel does not admit a Kraus representation with tensor product structure.}
\end{property}

\begin{property}
    The introduced measures are invariant under local unitary transformations, i.e., for $\Phi'=\Ad_{U_A\otimes U_B}\circ \Phi\circ \Ad_{V_A\otimes V_B}$,
\begin{equation}
e_C(\Phi')=e_C(\Phi),\qquad e_{\mathcal{N}}(\Phi')=e_{\mathcal{N}}(\Phi),
\end{equation}
since the product Haar measure and $(C,\mathcal{N})$ are invariant under local unitaries.
\end{property}

\begin{property}\label{prop:convexity-monotonicity}
    $e_{C}$ and $e_{\mathcal N}$ are convex, while $e_{L}$ is concave:
\begin{align}
e_C(p\Phi+(1-p)\Psi)
&\le p\,e_C(\Phi)+(1-p)\,e_C(\Psi),
\label{eq:eC-convex}
\\
e_{\mathcal N}(p\Phi+(1-p)\Psi)
&\le p\,e_{\mathcal N}(\Phi)+(1-p)\,e_{\mathcal N}(\Psi),
\label{eq:eN-convex}
\\
e_L(p\Phi+(1-p)\Psi)
&\ge p\,e_L(\Phi)+(1-p)\,e_L(\Psi).
\label{eq:eL-concave}
\end{align}
\end{property}

\begin{proof}
    Concurrence and negativity are convex state entanglement monotones \cite{Wootters1998,VidalWerner2002}; hence
\begin{align*}
C\!\left(p\rho+(1-p)\sigma\right)
&\le p\,C(\rho)+(1-p)\,C(\sigma),
\\
\mathcal N\!\left(p\rho+(1-p)\sigma\right)
&\le p\,\mathcal N(\rho)+(1-p)\,\mathcal N(\sigma).
\end{align*}
Applying this pointwise to
\begin{align*}
&(p\Phi+(1-p)\Psi)(\ketbra{\psi_A\psi_B}{\psi_A\psi_B})
\\&=
p\,\Phi(\ketbra{\psi_A\psi_B}{\psi_A\psi_B})
+(1-p)\,\Psi(\ketbra{\psi_A\psi_B}{\psi_A\psi_B})
\end{align*}
and averaging proves Eqs.~\eqref{eq:eC-convex} and \eqref{eq:eN-convex}. On the other hand, $E_L(\rho)=1-\Tr(\rho_A^2)$ is concave because the purity $\Tr(\rho_A^2)$ is convex. Therefore,
\begin{equation*}
E_L(p\rho+(1-p)\sigma)\ge p\,E_L(\rho)+(1-p)\,E_L(\sigma),
\end{equation*}
and averaging gives Eq.~\eqref{eq:eL-concave}.
\end{proof}

\begin{property}
    $e_{C}$ and $e_{\mathcal N}$ are monotone under local postprocessing. For any local channel
\begin{equation}
\Lambda_{\rm loc}=\Lambda_A\otimes \Lambda_B,
\end{equation}
one has
\begin{equation}
e_C(\Lambda_{\rm loc}\circ \Phi)\le e_C(\Phi),
\qquad
e_{\mathcal N}(\Lambda_{\rm loc}\circ \Phi)\le e_{\mathcal N}(\Phi).
\label{eq:local-post-monotone}
\end{equation}
No analogous monotonicity holds for $e_L$ in general.
\end{property}
\begin{proof}
	Let us denote
    \begin{equation*}
\rho_{\psi_A,\psi_B}:=\Phi(\ketbra{\psi_A\psi_B}{\psi_A\psi_B}).
\end{equation*}
Since concurrence and negativity do not increase under local CPTP maps \cite{Wootters1998,VidalWerner2002},
\begin{equation*}
C\!\left((\Lambda_A\otimes\Lambda_B)(\rho_{\psi_A,\psi_B})\right)\le C(\rho_{\psi_A,\psi_B}),
\end{equation*}
and similarly for $\mathcal N$. Taking an average proves Eqs.~\eqref{eq:local-post-monotone}. For $e_L$, such monotonicity fails: indeed $e_L(\mathcal I)=0$, whereas local dephasing $(\Lambda_\gamma\otimes\Lambda_\gamma)\circ \mathcal I$ has $e_L>0$ for $0<\gamma<1$, as shown in Sec.~\ref{sec:case-studies}.
\end{proof}

Property~\ref{prop:convexity-monotonicity} shows the basic conceptual difference between a linear-entropy-based definition of entangling power and the proposed quantifiers of entangling power. \rd{Convexity means that the entangling power of a mixture cannot exceed the corresponding weighted average of the component values.} 
The linear-entropy quantity behaves oppositely because it rewards loss of local purity. In this sense, $e_L$ is not merely quantitatively different from $e_C$ and $e_{\mathcal N}$, but it is qualitatively sensitive to a different physical property.

Finally, it is worth emphasizing the role of the two choices $C$ and $\mathcal N$. Concurrence is especially convenient for two qubits because it is both computable and tightly connected to the pure-state linear entropy. Negativity is often more conservative numerically, but has the advantage of being directly defined from the partial transpose and therefore conceptually close to the PPT criterion \cite{Peres1996,Horodecki1996,VidalWerner2002}.

\rd{The proposed product-input entangling powers should be distinguished from entanglement and entropy of the Choi-Jamio\l{}kowski state. Let $A_0B_0$ be the channel input, $AB$ its output, and $A'B'$ reference systems. For normalized maximally entangled states $\ket{\Omega_A}$ on $A_0A'$ and $\ket{\Omega_B}$ on $B_0B'$, define}
\begin{equation}
\rd{\omega_\Phi=(\Phi_{A_0B_0\rightarrow AB}\otimes\mathrm{id}_{A'B'})(\ketbra{\Omega_A}{\Omega_A}\otimes\ketbra{\Omega_B}{\Omega_B}).}
\label{eq:normalized-choi-minimal}
\end{equation}
\rd{The partition relevant to separability of the bipartite channel is $(AA'):(BB')$. A channel is separable in the sense of Eq.~\eqref{eq:sepchannel} if and only if $\omega_\Phi$ is separable across this partition \cite{Gour2021}. By contrast, $e_C$ and $e_{\mathcal N}$ quantify entanglement on the output partition $A:B$ generated from ancilla-free product inputs. The output corresponding to a product input is a local conditional slice of the Choi state,}
\begin{equation}
\rd{\Phi(\ketbra{\psi_A\psi_B}{\psi_A\psi_B})=d_A d_B\,{}_{A'B'}\!\bra{\psi_A^*\psi_B^*}\omega_\Phi\ket{\psi_A^*\psi_B^*}_{A'B'},}
\label{eq:choi-product-slice-minimal}
\end{equation}
\rd{where complex conjugation is taken in the basis defining the maximally entangled states. Applying the corresponding local covariant measurements to $A'$ and $B'$ and using strong monotonicity of negativity gives the following bound:}
\begin{equation}
\rd{e_{\mathcal N}(\Phi)\leq \mathcal N_{AA':BB'}(\omega_\Phi).}
\label{eq:choi-negativity-bound-minimal}
\end{equation}
\rd{Thus, entanglement of the Choi state provides an upper bound for the average output negativity for product inputs, but the converse implication fails, as the \textsc{swap} example shows. Choi-state entropy and related exchange-entropy bounds characterize channel mixedness rather than product-input entanglement generation \cite{Roga2010,Roga2011}.}

\section{Diagnostic regimes in representative two-qubit channels}
\label{sec:case-studies}
This section illustrates how the preceding diagnostics capture different
features of bipartite channels.
Throughout this section we restrict to two-qubit channels. 
\rd{Unless stated otherwise, the quantity $\chi_F$ is computed relative to the identity channel.}

\subsection{Nonentangling channels: False positives of \(e_L\)}
\label{subsec:false-positive-regime}

We consider channels that do not generate entanglement from product
inputs,
\begin{equation}
	e_C(\Phi)=e_{\mathcal N}(\Phi)=0,
\end{equation}
although the linear-entropy-based entangling power \(e_L\) defined in Ref.~\cite{Kong2015} remains positive. 
Thus, when \(e_L\) is interpreted as a diagnostic of genuine
mixed-state entanglement generation, it gives a false positive. The reason is
that for mixed outputs \(E_L\) measures local impurity rather than entanglement.

The correlated dephasing channel can be written as a mixture of product unitaries:
\begin{equation}
	\Phi_u(\rho)
	=
	(1-\gamma)\rho+\gamma (Z\otimes Z)\rho(Z\otimes Z),
	\quad
	u=1-2\gamma .
	\label{eq:globalPD_randomU}
\end{equation}
\begin{figure}[h]
	\centering
	\includegraphics[width=\columnwidth]{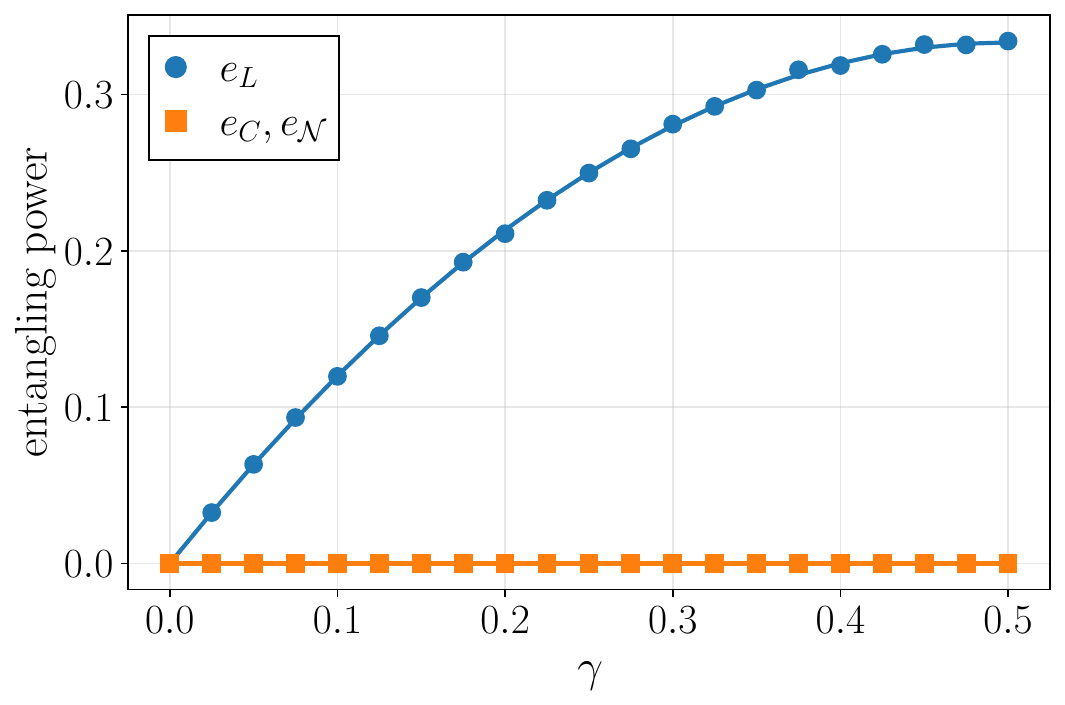}
	\caption{
		False-positive behavior of the linear-entropy quantity for the correlated
		dephasing channel \(\Phi_u\), with \(\gamma=(1-u)/2\) [see Eq. \eqref{eq:globalPD_randomU}]. The channel is a
		random-unitary mixture of product unitaries, so it cannot entangle product
		inputs and \(e_C=e_{\mathcal N}=0\). Nevertheless, \(e_L\) is positive because
		the output is locally mixed. Symbols denote numerical Haar averages over
		product inputs, while solid lines denote the analytic predictions.
	}
	\label{fig:correlated_dephasing}
\end{figure}
Equivalently, it dephases coherences between the two eigenspaces of
\(Z\otimes Z\). Since Eq.~\eqref{eq:globalPD_randomU} is a mixture of product
unitaries, it cannot entangle product inputs. Hence,
\begin{equation}
	e_C(\Phi_u)=e_{\mathcal N}(\Phi_u)=0.
\end{equation}
Nevertheless,
\begin{equation}
	e_L(\Phi_u)
	=
	\frac{4}{3}\gamma(1-\gamma)
	=
	\frac{1-u^2}{3},
\end{equation}
as shown in Fig.~\ref{fig:correlated_dephasing}. The same channel has
\begin{equation}\label{eq:fid-corr-deph}
	\Favg(\Phi_u)=\frac35+\frac25 u,
	\qquad
	\Fprod(\Phi_u)=\frac59+\frac49 u,
\end{equation}
and therefore
\begin{equation}\label{eq:chi-corr-deph}
	\chi_F(\Phi_u)
	=
	\frac{2}{45}(1-u)
	=
	\frac{4\gamma}{45}
	\geq0 .
\end{equation}
Thus, correlated dephasing is nonentangling on product inputs, but it preserves
lower-purity input orbits better than product inputs.

This behavior should be contrasted with local product noise. For independent
local phase flips,
\begin{equation}
	\Phi_\gamma^{Z,{\rm loc}}
	=
	\Lambda_\gamma\otimes\Lambda_\gamma,
	\qquad
	\Lambda_\gamma(\rho)=(1-\gamma)\rho+\gamma Z\rho Z,
\end{equation}
one obtains
\begin{align}
	e_C(\Phi_\gamma^{Z,{\rm loc}})
	&=
	e_{\mathcal N}(\Phi_\gamma^{Z,{\rm loc}})
	=
	0,\nonumber
	\\e_L(\Phi_\gamma^{Z,{\rm loc}})
	&=
	\frac{4}{3}\gamma(1-\gamma),
\end{align}
and
\begin{equation}
	\chi_F(\Phi_\gamma^{Z,{\rm loc}})
	=
	\frac{4\gamma(4\gamma-3)}{45}.
\end{equation}
For the physically relevant weak- and moderate-noise regime
and up to \(\gamma\le3/4\), this is nonpositive, in agreement with
Proposition~\ref{prop:product-channel-chi}. The sign change for
\(\gamma>3/4\) is not a contradiction: in this regime the local channel is
closer to a deterministic unitary phase flip than to the identity.

\begin{figure}[h]
	\centering
	\includegraphics[width=\columnwidth]{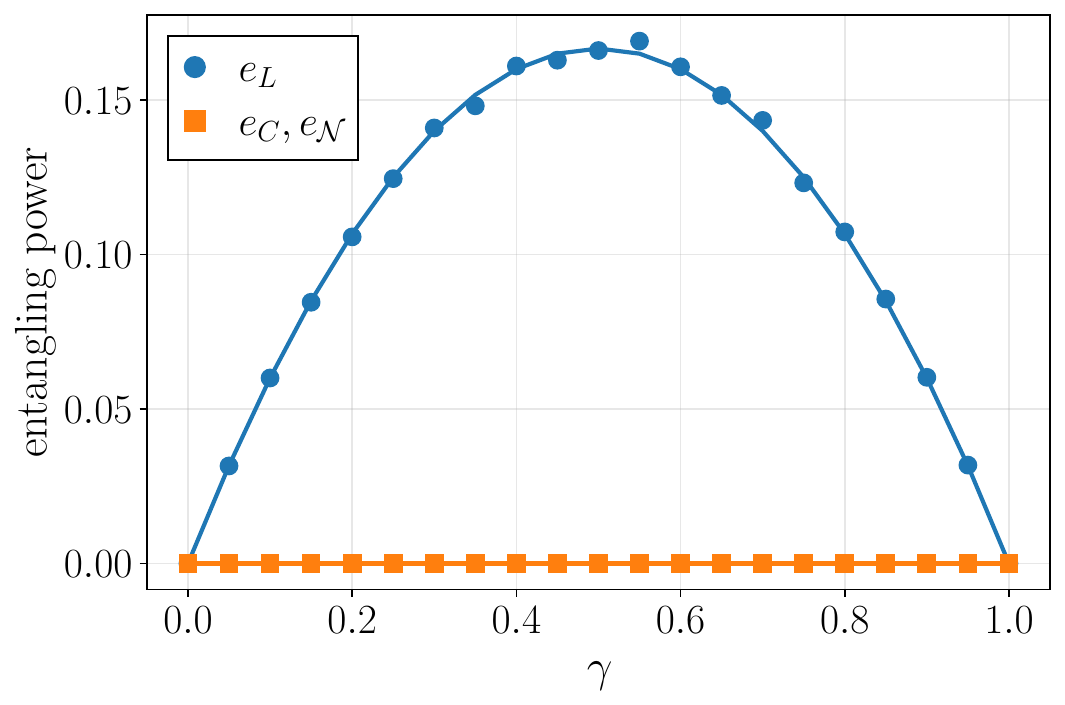}
		\caption{
			False-positive behavior of \(e_L\) for local amplitude damping
			\(\Phi_\gamma^{\rm AD}\). The channel is local; hence,
			\(e_C=e_{\mathcal N}=0\), but \(e_L=2\gamma(1-\gamma)/3\) is nonzero
			because it measures the local impurity of the mixed output.
		}
	\label{fig:local_amplitude_damping}
\end{figure}
For independent local amplitude damping (AD),
\begin{equation}
	\Phi_\gamma^{\rm AD}
	=
	\Lambda_\gamma^{\rm AD}\otimes \Lambda_\gamma^{\rm AD},
\end{equation}
with single-qubit Kraus operators
\begin{equation}
	A_0=
	\begin{pmatrix}
		1&0\\
		0&\sqrt{1-\gamma}
	\end{pmatrix},
	\qquad
	A_1=
	\begin{pmatrix}
		0&\sqrt{\gamma}\\
		0&0
	\end{pmatrix},
\end{equation}
we define
\begin{equation}
	q(\gamma):=1+\sqrt{1-\gamma}.
\end{equation}
Then
\begin{equation}
	e_C(\Phi_\gamma^{\rm AD})
	=
	e_{\mathcal N}(\Phi_\gamma^{\rm AD})
	=
	0,
	\qquad
	e_L(\Phi_\gamma^{\rm AD})
	=
	\frac{2}{3}\gamma(1-\gamma),
\end{equation}
and
\begin{equation}
	\chi_F(\Phi_\gamma^{\rm AD})
	=
	-
	\frac{(q(\gamma)^2-1)(4-q(\gamma)^2)}{45}
	\le0.
\end{equation}
Thus, local amplitude damping is nonentangling on product inputs and favors
productlike input states at the level of the fidelity-bias diagnostic (see
Fig.~\ref{fig:local_amplitude_damping}).

\begin{figure}[h]
	\centering
	\includegraphics[width=\columnwidth]{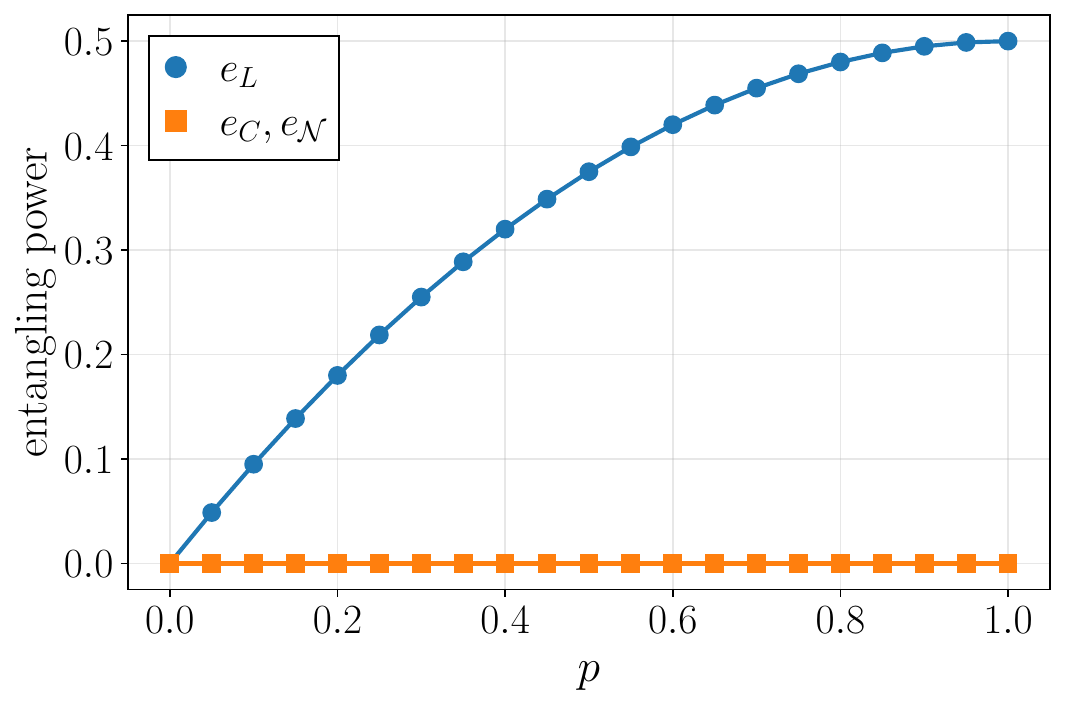}
		\caption{
			False-positive behavior of \(e_L\) for local depolarizing noise
			\(\mathcal D_p\otimes\mathcal D_p\). Proper entangling powers vanish for all
			\(p\), while \(e_L=p-p^2/2\) increases with the local mixedness of the output.
		}
	\label{fig:local_depolarizing_noise}
\end{figure}
For local depolarizing noise,
\begin{equation}
	\Phi_p^{\rm loc}
	=
	\mathcal D_p\otimes\mathcal D_p,
	\qquad
	\mathcal D_p(\rho)=(1-p)\rho+p\frac{\id}{2},
\end{equation}
one has
\begin{equation}
	e_C(\Phi_p^{\rm loc})
	=
	e_{\mathcal N}(\Phi_p^{\rm loc})
	=
	0,
	\qquad
	e_L(\Phi_p^{\rm loc})
	=
	p-\frac{p^2}{2},
\end{equation}
and
\begin{equation}
	\chi_F(\Phi_p^{\rm loc})
	=
	\frac{p(p-1)}{5}
	\le0 .
\end{equation}
This gives another local-noise example in which \(e_L\) is positive although no
entanglement is generated from product inputs (see
Fig.~\ref{fig:local_depolarizing_noise}).

Finally, the global depolarizing channel
\begin{equation}
	\Phi_p^{\rm glob}(\rho)
	=
	(1-p)\rho+p\frac{\id}{4}
\end{equation}
also has
\begin{equation}
	e_C(\Phi_p^{\rm glob})
	=
	e_{\mathcal N}(\Phi_p^{\rm glob})
	=
	0.
\end{equation}
However,
\begin{equation}
	\Favg(\Phi_p^{\rm glob})
	=
	\Fprod(\Phi_p^{\rm glob})
	=
	1-\frac{3p}{4},
\end{equation}
so
\begin{equation}
	\chi_F(\Phi_p^{\rm glob})=0.
\end{equation}
Thus, global depolarization is purity blind in the sense of the orbit-averaged
fidelity diagnostic.

\subsection{Coherent entangling dynamics}
\label{subsec:coherent-entangling-dynamics}

We next consider a genuinely entangling unitary family. The controlled-phase
gate
\begin{equation}
	\mathrm{\textsc{cp}}(\varphi)
	=
	\mathrm{diag}(1,1,1,e^{i\varphi})
	\label{eq:controlled-phase-gate}
\end{equation}
acts on a product input
\[
\ket{\psi_A}=a\ket0+b\ket1,
\qquad
\ket{\psi_B}=c\ket0+d\ket1,
\]
as
\begin{equation}
	\ket{\Psi_{\rm out}}
	=
	ac\ket{00}
	+
	ad\ket{01}
	+
	bc\ket{10}
	+
	e^{i\varphi}bd\ket{11}.
\end{equation}
The output concurrence is
\begin{equation}
	C
	=
	4|abcd|\left|\sin\frac{\varphi}{2}\right|.
\end{equation}
Averaging over product Haar inputs gives the following entangling power values for the controlled-phase (\textsc{cp}) unitary channel \(\Ad_{\mathrm{\textsc{cp}}(\varphi)}\):
\begin{align}
	e_C(\Ad_{\mathrm{\textsc{cp}}(\varphi)})
	&=
	\frac{\pi^2}{16}
	\left|\sin\frac{\varphi}{2}\right|,\label{eq:CP_C}
	\\
	e_{\mathcal N}(\Ad_{\mathrm{\textsc{cp}}(\varphi)})
	&=
	\frac{\pi^2}{32}
	\left|\sin\frac{\varphi}{2}\right|,\label{eq:CP_N}
	\\
	e_L(\Ad_{\mathrm{\textsc{cp}}(\varphi)})
	&=
	\frac{2}{9}\sin^2\frac{\varphi}{2}.\label{eq:CP_L}
\end{align}
Here the outputs are pure, so \(e_{\mathcal N}=e_C/2\) pointwise and
\(e_L\) is directly related to \(C^2\). This is why all three quantities vanish
and peak at the same values of \(\varphi\), as shown in
Fig.~\ref{fig:control_phase}. For $\varphi=\pi$, the controlled-phase gate reduces to a controlled-Z (\textsc{cz}) gate. 

\begin{figure}[h]
	\centering
	\includegraphics[width=\columnwidth]{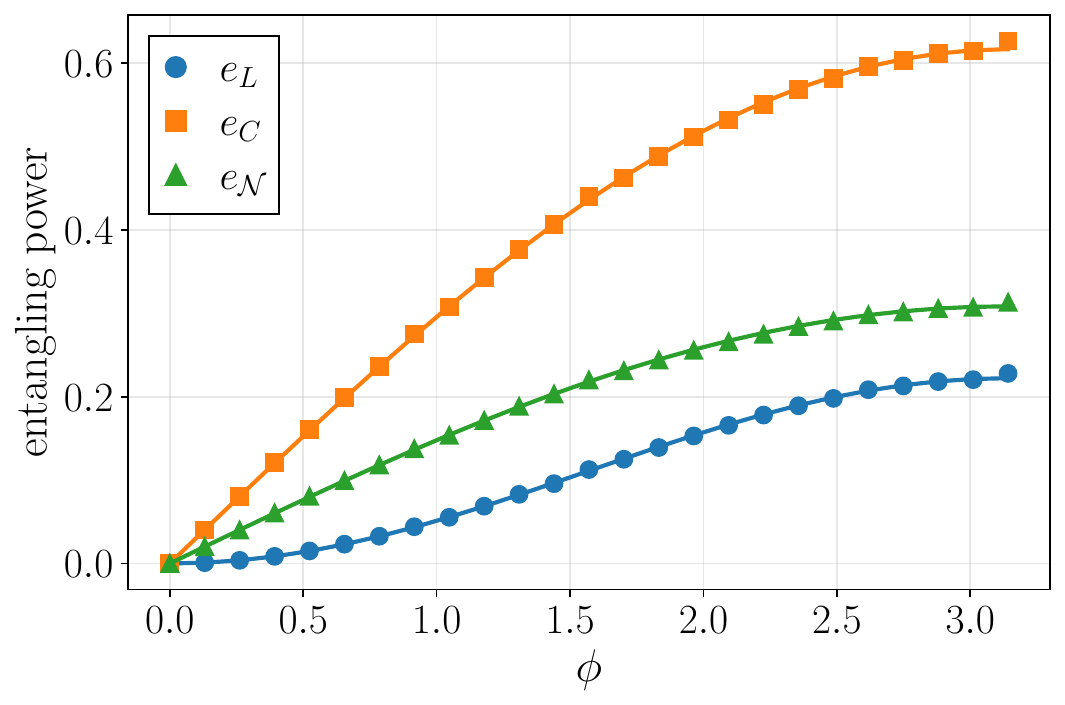}
	\caption{
		Controlled-phase unitary benchmark. Since the output is pure,
		\(e_{\mathcal N}=e_C/2\) and \(e_L\) is proportional to the average squared
		concurrence. The concurrence- and negativity-based entangling powers scale as
		\(\left|\sin(\varphi/2)\right|\), whereas \(e_L\) scales as
		\(\sin^2(\varphi/2)\). Symbols denote numerical Haar averages and solid lines
		denote analytic expressions.
	}
	\label{fig:control_phase}
\end{figure}

The fidelity bias for the controlled-phase gate reads
\begin{equation}
	\chi_F(\Ad_{\mathrm{\textsc{cp}}(\varphi)})
	=
	-\frac{2}{45}\sin^2\frac{\varphi}{2}
	\le0.
\end{equation}
\rd{This example should not be interpreted as a relation between the two diagnostics. The positive entangling power quantifies output entanglement generated from product inputs, whereas the negative identity-relative $\chi_F$ only states that lower-purity input orbits have a smaller average overlap with their initial states.}

\subsection{Noisy entangling gate: CZ gate under correlated dephasing}
\label{subsec:CZ_corr_deph}

The previous two regimes meet in a noisy entangling gate. Consider an ideal \textsc{cz}
gate followed by correlated dephasing:
\begin{equation}
	\Phi_u^{\rm \textsc{cz}}
	=
	\Phi_u\circ\Ad_{\mathrm{\textsc{cz}}},
	\qquad
	\mathrm{\textsc{cz}}=\mathrm{diag}(1,1,1,-1).
	\label{eq:cz-corr-deph}
\end{equation}
Since both the ideal gate and the noise are diagonal in the computational basis,
they commute. Therefore, the relative-fidelity quantities with respect to the
target \(\Ad_{\mathrm{\textsc{cz}}}\) are exactly those of the correlated dephasing error [see Eqs.~\eqref{eq:fid-corr-deph} and \eqref{eq:chi-corr-deph}].
Thus, at the level of the relative-fidelity diagnostic with respect to the
ideal \textsc{cz} target, the error favors lower-purity input orbits in the same sense as the correlated dephasing channel itself.

\begin{figure}[h]
	\centering
	\includegraphics[width=\columnwidth]{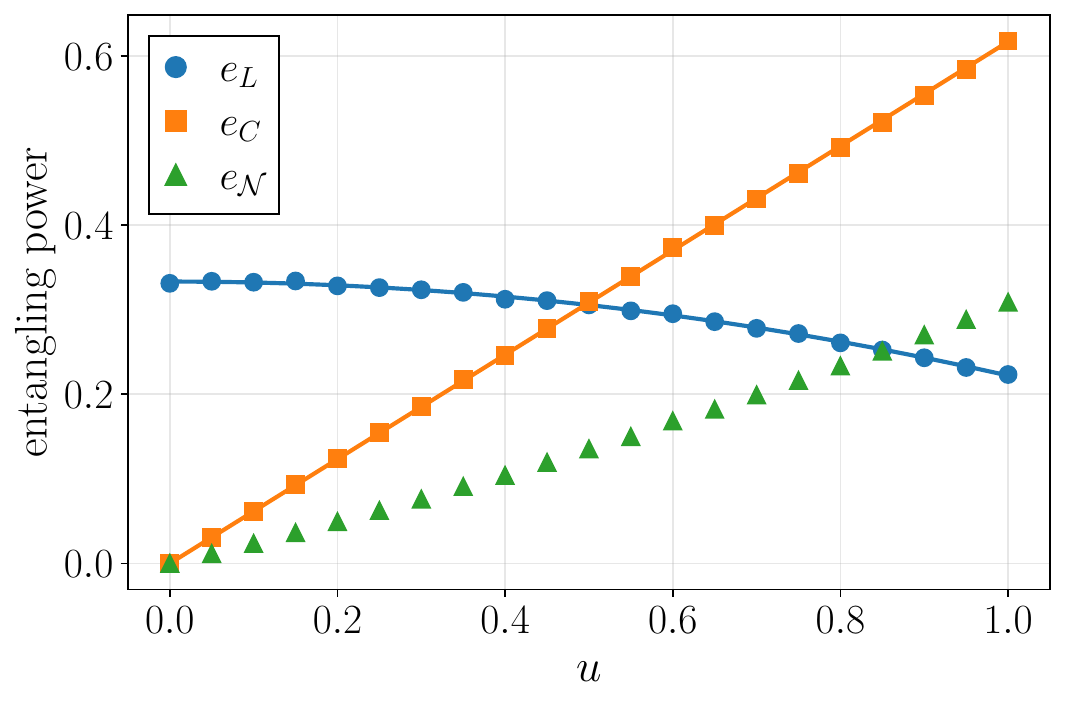}
		\caption{Entangling  \textsc{cz} gate with noise represented by correlated dephasing. The
			concurrence-based entangling power scales linearly,
			\(e_C(\Phi_u^{\rm \textsc{cz}})=u\pi^2/16\). The negativity-based entangling power \(e_{\mathcal N}\) is
			estimated numerically.
			The linear-entropy-based quantity \(e_L\) is shown for comparison.
		}
	\label{fig:cz-correlated-dephasing}
\end{figure}

At the same time, the channel remains entangling for \(u>0\). Correlated
dephasing rescales the output concurrence pointwise by \(u\), as shown in
Appendix~\ref{app:concurrence-rescaling}. Hence,
\begin{equation}
	e_C(\Phi_u^{\rm \textsc{cz}})
	=
	u\,e_C(\Ad_{\mathrm{\textsc{cz}}})
	=
	u\frac{\pi^2}{16}.
\end{equation}
The linear-entropy quantity is
\begin{equation}
	e_L(\Phi_u^{\rm \textsc{cz}})
	=
	\frac13-\frac{u^2}{9}.
\end{equation}
Unlike in the noiseless controlled-phase case, the outputs are now generally mixed. Therefore, \(e_L\) no longer measures entanglement alone, but it also contains contributions from output mixedness. The negativity-based entangling power remains a proper mixed-state entanglement diagnostic and is shown numerically in Fig.~\ref{fig:cz-correlated-dephasing}.

Combining the concurrence formula with the relative average fidelity gives
\begin{equation}
	e_C(\Phi_u^{\rm \textsc{cz}})
	=
	\frac{\pi^2}{32}
	\left(
	5\Favg(\Phi_u^{\rm \textsc{cz}},\Ad_{\mathrm{\textsc{cz}}})-3
	\right),
\end{equation}
\rd{for the one-parameter correlated-dephasing model.}
The same parameter
\(u\) controls both the relative fidelity and the concurrence damping.

\begin{figure}[h]
	\centering
	\includegraphics[width=\columnwidth]{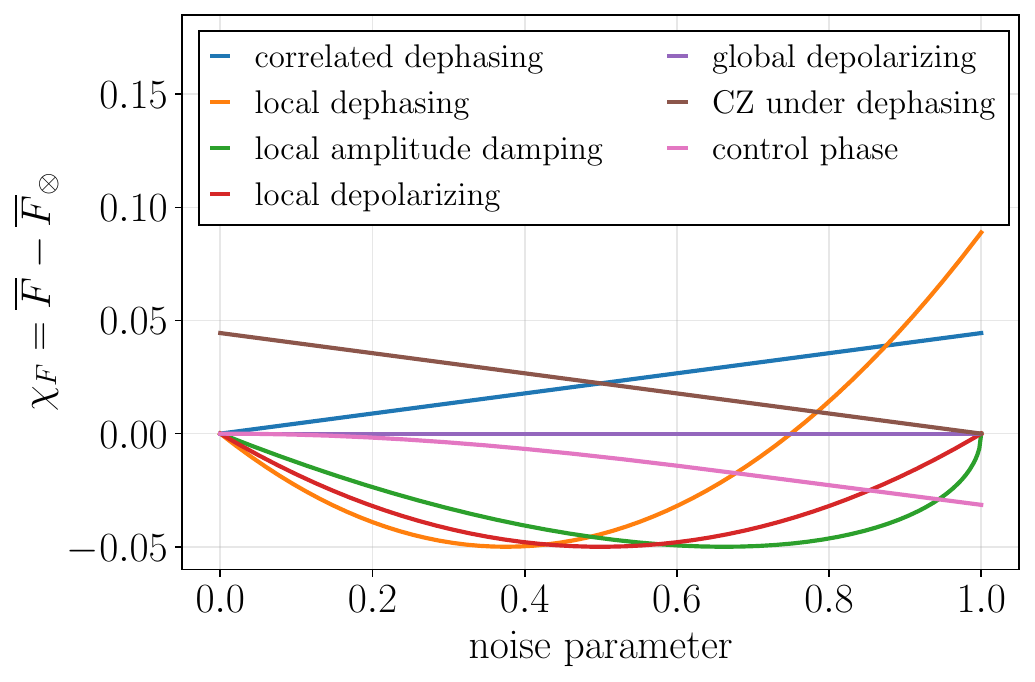}
	\caption{
		Fidelity-bias parameter \(\chi_F\) for the representative models.
		Positive values indicate that lower-purity input orbits are better preserved on
		average than product inputs, while negative values indicate the opposite.
		The local-noise models are nonpositive in the usual weak-noise regime,
		correlated dephasing has \(\chi_F>0\), and global depolarization is
		purity blind with \(\chi_F=0\). For the noisy-\textsc{cz} curve, \(\chi_F\) is computed relative to the ideal \textsc{cz} target;
		for the other noise channels it is computed relative to the identity channel.
	}
	\label{fig:chi}
\end{figure}
Figure~\ref{fig:chi} summarizes the fidelity-bias parameter for all discussed models. 
\rd{It also provides a concrete illustration of the independent information carried by the two diagnostics: correlated and local dephasing both have zero proper entangling power, while their fidelity biases can have different signs. According to Eq.~\eqref{eq:orbit-theorem-qudit}, this sign difference has a direct operational
meaning. For correlated dephasing, product-input fidelity underestimates the
average fidelity of lower-purity input orbits, whereas for independent local
dephasing in the weak- and moderate-noise regime it overestimates that
fidelity. Both channels have vanishing proper entangling
power.}

\section{Bounds for concurrence-based entangling power}
\label{sec:bounds}

The concurrence-based entangling power \(e_C(\Phi)\) is a genuine
entanglement-generation diagnostic, but it is generally difficult to evaluate analytically. 
The difficulty comes from the nonlinear dependence of concurrence on the output state. In this section we derive computable analytic bounds on \(e_C(\Phi)\). The main idea is to use the squared concurrence, called the tangle, $\tau(\rho):=C^2(\rho)$ \cite{Wootters1998}, as an auxiliary quantity for that purpose.
The corresponding concurrence second moment is
\begin{equation}
	e_{C^2}(\Phi)
	:=
	\int d\psi_A d\psi_B\,
	C^2\!\left(
	\Phi(\ketbra{\psi_A\psi_B}{\psi_A\psi_B})
	\right).
	\label{eq:etau}
\end{equation}
Since \(0\le C(\rho)\le 1\) for two qubits, one has pointwise
\(C(\rho)^2\le C(\rho)\). On the other hand, the function
\(x\mapsto x^2\) is convex, so Jensen's inequality gives
\[
\left(\int d\psi_A d\psi_B\, C(\rho_{\psi_A,\psi_B})\right)^2
\le
\int d\psi_A d\psi_B\, C(\rho_{\psi_A,\psi_B})^2,
\]
where
\[
\rho_{\psi_A,\psi_B}
:=
\Phi(\ketbra{\psi_A\psi_B}{\psi_A\psi_B}) .
\]
Therefore, we arrive at the bounds
\begin{equation}
	e_C(\Phi)^2
	\le
	e_{C^2}(\Phi)
	\le
	e_C(\Phi).
	\label{eq:eC-etau-relations}
\end{equation}
Thus, lower bounds on \(e_{C^2}\) also lower-bound \(e_C\), while upper bounds on
\(e_{C^2}\) give upper bounds on \(e_C\) after taking a square root.

We now use the known two-qubit bounds on the tangle \cite{Mintert07,Zhang08},
\begin{equation}
	2\delta_P(\rho)
	\le
	\tau(\rho)
	\le
	2E_L(\rho),
	\label{eq:state-tangle-bounds}
\end{equation}
where
\begin{equation}
	\delta_P(\rho_{AB}):=
	\Tr(\rho_{AB}^2)-\Tr(\rho_A^2),
	\qquad
	\rho_A=\Tr_B\rho_{AB}.
\end{equation}
The quantity \(\delta_P(\rho)\) may be negative. 
Note that the same construction may be repeated with \(A\) and \(B\) interchanged. Thus,
for asymmetric channels one may use the alternative that provides the stronger lower bound.

We recall that $\tau(\rho)=C^2(\rho)$. After averaging over pure
product inputs, we obtain
\begin{equation}
	2\overline{\delta_P}^{\otimes}(\Phi)
	\le
	e_{C^2}(\Phi)
	\le
	2e_L(\Phi),
	\label{eq:etau-bounds}
\end{equation}
where
\begin{equation}
	\overline{\delta_P}^{\otimes}(\Phi)
	:=
	\int d\psi_A d\psi_B\,
	\delta_P\!\left(
	\Phi(\ketbra{\psi_A\psi_B}{\psi_A\psi_B})
	\right).
	\label{eq:def-average-delta-product}
\end{equation}
Combining Eqs.~\eqref{eq:eC-etau-relations} and
\eqref{eq:etau-bounds} gives direct bounds on the concurrence-based entangling
power:
\begin{equation}
	2\max\left\{0,\overline{\delta_P}^{\otimes}(\Phi)\right\}
	\le
	e_C(\Phi)
	\le
	\sqrt{2e_L(\Phi)}.
	\label{eq:eC-direct-bounds}
\end{equation}
The lower bound is shown to be a useful certificate of entanglement generation and it can give a tighter estimate than the upper bound in parts of the parameter range.
If \(\overline{\delta_P}^{\otimes}(\Phi)>0,\) then \(e_C(\Phi)>0\). Since, for two-qubit states, concurrence and negativity
vanish on exactly the same separable states, this also certifies
\(e_{\mathcal N}(\Phi)>0\). However, the present purity-difference method does
not provide a comparable quantitative lower bound on \(e_{\mathcal N}\).
Using \(2\mathcal N(\rho)\le C(\rho)\), it does give the following upper bound:
\begin{equation}
	e_{\mathcal N}(\Phi)
	\le
	\frac12 e_C(\Phi)
	\le
	\frac12\sqrt{2e_L(\Phi)}.
	\label{eq:eN-upper-from-eL}
\end{equation}

It remains to compute
\(\overline{\delta_P}^{\otimes}(\Phi)\). This can be done analytically by
using a two-copy representation. Let \(S_{AA'}\) denote the \textsc{swap} operator between subsystem \(A\) and its copy \(A'\). It acts as the identity on the remainder of the full Hilbert space. \(S_{BB'}\) is defined analogously. For product Haar inputs, define \cite{Zanardi2000, CollinsSniady2006}
\begin{equation}
	\begin{split}
		\Omega_{\otimes}
		&:=
		\int d\psi_A d\psi_B\,
		\left(
		\ketbra{\psi_A}{\psi_A}
		\otimes
		\ketbra{\psi_B}{\psi_B}
		\right)^{\otimes 2}
		\\
		&=
		\frac{1}{d^2(d+1)^2}
		\left(\mathbbm 1+S_{AA'}\right)
		\left(\mathbbm 1+S_{BB'}\right).
	\end{split}
	\label{eq:Omega-product}
\end{equation}

The bounds above are used here for two-qubit channels. The following two-copy
formula for \(\overline{\delta_P}^{\otimes}\), however, is purely algebraic and
holds for equal local dimension \(d\).

\begin{proposition}
	\label{prop:delta-P-tensor}
	Let
	\[
	\Phi(\rho)=\sum_\alpha K_\alpha\rho K_\alpha^\dagger
	\]
	be a CPTP map on \(\mathbb C^d\otimes\mathbb C^d\). Then
	\begin{align}
		\overline{\delta_P}^{\otimes}(\Phi)=
		\sum_{\alpha,\beta}
		\Tr \big[
		(K_\alpha\otimes K_\beta)
		\Omega_{\otimes} (K_\alpha^\dagger\otimes K_\beta^\dagger)\nonumber
		\\ \times\left(S_{AA'}S_{BB'}-S_{AA'}\right)\big].
		\label{eq:av-delta-p-prod}
	\end{align}
\end{proposition}

\begin{proof}
	For any bipartite state \(\rho_{AB}\), the \textsc{swap} trick gives
	\begin{equation}
		\Tr(\rho_{AB}^2)
		=
		\Tr\!\left[
		(\rho_{AB}\otimes\rho_{AB})S_{AA'}S_{BB'}
		\right],
	\end{equation}
	and
	\begin{equation}
		\Tr(\rho_A^2)
		=
		\Tr\!\left[
		(\rho_{AB}\otimes\rho_{AB})S_{AA'}
		\right].
	\end{equation}
	Hence,
	\begin{equation}
		\delta_P(\rho_{AB})
		=
		\Tr\!\left[
		(\rho_{AB}\otimes\rho_{AB})
		\left(S_{AA'}S_{BB'}-S_{AA'}\right)
		\right].
	\end{equation}
	Substituting
	\[
	\rho_{AB}
	=
	\Phi(\ketbra{\psi_A\psi_B}{\psi_A\psi_B})
	=
	\sum_\alpha
	K_\alpha
	\ketbra{\psi_A\psi_B}{\psi_A\psi_B}
	K_\alpha^\dagger
	\]
	and averaging over product Haar inputs gives Eq.~\eqref{eq:av-delta-p-prod},
	because the averaged two-copy input state is precisely
	\(\Omega_{\otimes}\).
\end{proof}

In Fig.~\ref{fig:lower_bounds}, we show a comparison of the concurrence-based entangling power $e_C$ and its bounds, as given by Eq. (\ref{eq:eC-direct-bounds}),  for different noise models, which we discuss below.

We first consider a \textit{mixed-unitary channel} with Kraus operators
\(K_0=\sqrt p\,U_1\) and \(K_1=\sqrt{1-p}\,U_2\), where \(U_1\) and \(U_2\)
are Haar-random unitaries in \(U(4)\), and \(p\in[0,1]\).
The detailed dependence on \(p\) depends on the chosen
unitaries, but the bounds capture the suppression of
concurrence-based entangling power under incoherent mixing.

The \textit{phase-damping channel} applied to two-qubit systems corresponds to Kraus operators 
\begin{equation}
	K_0 = \mathrm{diag}(e^{-\Gamma t^2/2},1,1,e^{-\Gamma t^2/2}),
	\label{kraus_damping}
\end{equation}
and
\begin{equation}
	K_1=\mathrm{diag}(\sqrt{1-e^{-\Gamma t^2}},0,0,\sqrt{1-e^{-\Gamma t^2}}),
	\label{kraus_damping_2}
\end{equation}
where we set $\Gamma = 1$.
This is nothing but the correlated dephasing channel of Sec.~\ref{sec:case-studies} after the change of variable, $u=e^{-\Gamma t^2/2}$. As discussed above, this channel is separable and thus has vanishing entangling power. The lower bound correctly captures this characteristic.

The \textit{controlled-Z gate under phase damping} corresponds to
Kraus operators
\begin{equation}
	K^{\textsc{cz}}_0 = e^{-ig t}\mathrm{diag}(e^{-\Gamma t^2/2},1,1,e^{2ig t}e^{-\Gamma t^2/2}),
	\label{kraus_damping_cz}
\end{equation}
\begin{equation}
	K^{\textsc{cz}}_1=\mathrm{diag}(e^{-igt}\sqrt{1-e^{-\Gamma t^2}},0,0,e^{igt}\sqrt{1-e^{-\Gamma t^2}}),
	\label{kraus_damping_cz2}
\end{equation}
with parameters set as $g=3/2$ and $\Gamma=1$.
This channel describes an application of the unitary gate, 
\begin{equation}
	U_{g,t}=e^{-igt}\mathrm{diag}(1,1,1,e^{2igt}),
\end{equation}
which is not entangling for $gt=k\pi$ and it is maximally entangling for $gt=\pi/2+k\pi$, being equivalent to the \textsc{cz} gate—under the phase damping noise. 
For the parameters used in Fig.~\ref{fig:lower_bounds}, the noiseless maximum
near \(t=\pi/(2g)\) is reduced and shifted to smaller times by the damping.
 Interestingly, the lower bound is capable of describing the recession of the maximum due to the effect of noise.

\begin{figure*}[t]
	\centering
	\includegraphics[width=\textwidth]{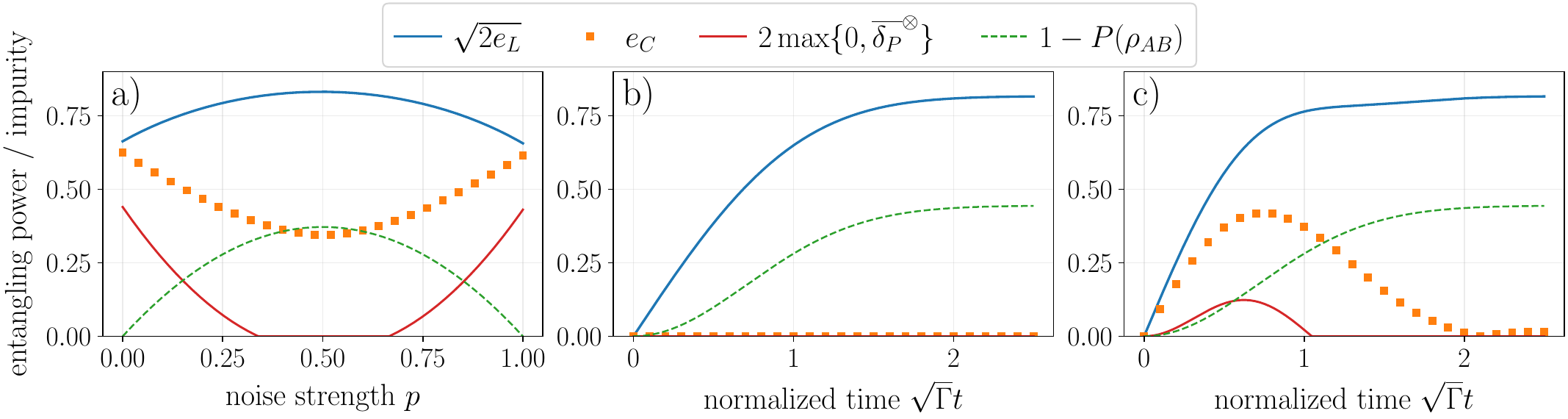}
		\caption{
		Comparison between the concurrence-based entangling power \(e_C\), computed numerically over Haar-random product inputs, and the analytic bounds from Eq.~\eqref{eq:eC-direct-bounds}. {The blue and red solid curves represent the upper and lower bounds, respectively.} The green dashed curve shows the averaged global output impurity.  (a)  A mixed-unitary channel;  (b)  correlated phase damping, for which the proper entangling power vanishes;  and (c)  controlled-phase evolution under correlated phase damping.	
	}
	\label{fig:lower_bounds}
\end{figure*}

\section{Average entanglement variation over fixed Schmidt orbits}
\label{sec:variation}

We now extend the product-input framework to initially entangled pure inputs.
For two qubits, we fix an input Schmidt angle \(\theta\) and average over the
corresponding local-unitary orbit. Namely, we consider states of the form
\begin{equation}
	\ket{\psi_\theta(U_A,U_B)}
	=
	(U_A\otimes U_B)
	\left(
	\cos\theta\,\ket{00}
	+
	\sin\theta\,\ket{11}
	\right).
	\label{eq:psi-theta-orbit}
\end{equation}
The product-input case is recovered at \(\theta=0\), while
\(\theta=\pi/4\) corresponds to maximally entangled inputs.

For an entanglement measure \(M\), define the orbit-averaged output
entanglement
\begin{align}
	\overline{M}_{\rm out}(\Phi,\theta)
	&:=
	\int dU_A dU_B\,
	\nonumber\\&\times M\!\left(
	\Phi(\ketbra{\psi_\theta(U_A,U_B)}{\psi_\theta(U_A,U_B)})
	\right).
	\label{eq:orbit-average-M-out}
\end{align}
The corresponding average entanglement variation is
\begin{equation}
	\Delta e_M(\Phi,\theta)
	:=
	\overline{M}_{\rm out}(\Phi,\theta)
	-
	M_{\rm in}(\theta),
	\label{eq:def-general-ent-variation}
\end{equation}
where
\[
M_{\rm in}(\theta)
=
M\!\left(
\ketbra{\psi_\theta(U_A,U_B)}{\psi_\theta(U_A,U_B)}
\right)
\]
is independent of \(U_A,U_B\).
Equation (\ref{eq:orbit-average-M-out}) is immediately extended to a pair of qudits by using $\vec \lambda$ in place of $\theta$.

For the initial pure state in Eq.~\eqref{eq:psi-theta-orbit}, one has
concurrence \(C_{\rm in}(\theta)=\sin 2\theta\), and negativity	\(\mathcal N_{\rm in}(\theta)=
	\frac12\sin 2\theta\). 
The concurrence- and negativity-based entanglement variations are
\begin{align}
	\Delta e_C(\Phi,\theta)
	&=
	\overline{C}_{\rm out}(\Phi,\theta)
	-
	\sin 2\theta,
	\label{eq:def-delta-eC-theta}
	\\
	\Delta e_{\mathcal N}(\Phi,\theta)
	&=
	\overline{\mathcal N}_{\rm out}(\Phi,\theta)
	-
	\frac12\sin 2\theta.
	\label{eq:def-delta-eN-theta}
\end{align}
They measure the average change of genuine two-qubit entanglement under the
channel. Positive values indicate average entanglement enhancement on the
chosen fixed-Schmidt-spectrum orbit, while negative values indicate average
entanglement degradation.

Using the tangle \(
\tau(\rho)=C(\rho)^2
\)
as an auxiliary squared-concurrence moment, we define
\begin{align}
	\overline{C^2}_{\rm out}(\Phi,\theta)
	&
	:=
	\int dU_A dU_B\,
	\nonumber\\&\times C^2\left(
	\Phi(\ketbra{\psi_\theta(U_A,U_B)}{\psi_\theta(U_A,U_B)})
	\right)
	\label{eq:def-average-tau-out-theta}
\end{align}
and
\begin{equation}
	\Delta e_{C^2}(\Phi,\theta)
	:=
	\overline{C^2}_{\rm out}(\Phi,\theta)
	-
	\sin^2 2\theta.
	\label{eq:def-delta-etau-theta}
\end{equation}
This quantity is useful because \(C^2\) admits computable purity bounds.
We also use the linear-entropy variation
\begin{equation}
	\Delta e_L(\Phi,\theta)
	:=
	\overline{E_L}_{\rm out}(\Phi,\theta)
	-
	E_{L,\rm in}(\theta),
	\label{eq:def-delta-eL-theta}
\end{equation}
where
\begin{align}
	\overline{E_L}_{\rm out}(\Phi,\theta)
	&:=
	\int dU_A dU_B\,
	\nonumber\\&\times E_L\!\left(
	\Phi(\ketbra{\psi_\theta(U_A,U_B)}{\psi_\theta(U_A,U_B)})
	\right).
\end{align}
Unlike \(\Delta e_C\), \(\Delta e_N\), and \(\Delta e_{C^2}\), the quantity \(\Delta e_L\) is not based on a genuine mixed-state entanglement monotone. Nevertheless, it yields the analytic upper bound \(\Delta e_{C^2}(\Phi,\theta)\leq 2 \Delta e_L (\Phi,\theta)\).

\begin{figure*}[t]
	\centering
	\begin{subfigure}{1\textwidth}
		\centering
		\includegraphics[width=\linewidth]{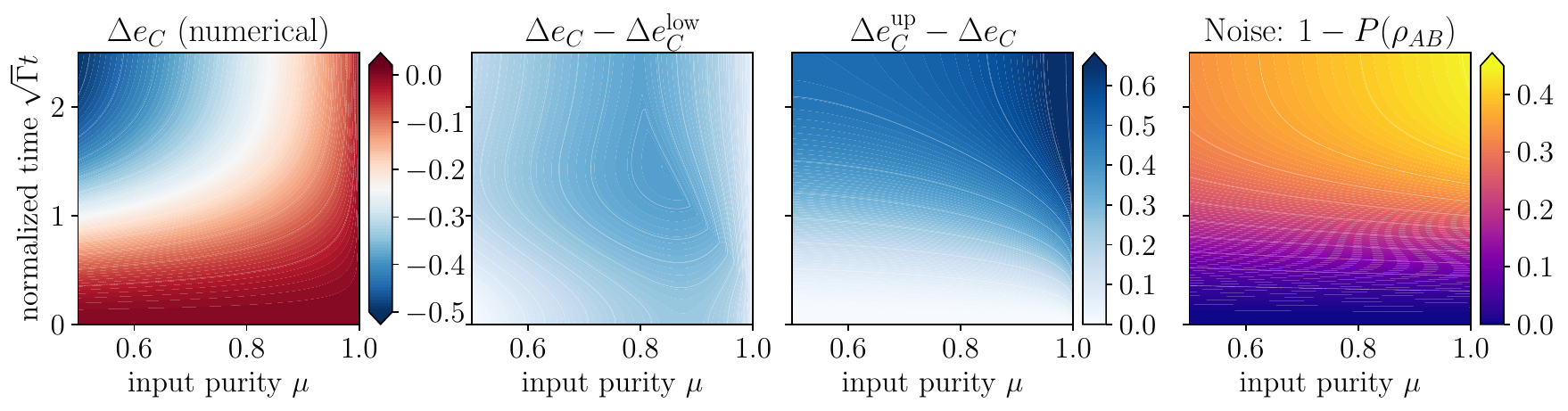}
	\end{subfigure}
	\caption{
		Average concurrence variation for the correlated phase-damping channel,	Eqs.~\eqref{kraus_damping} and \eqref{kraus_damping_2}, shown together with the tightness gaps of the
		analytic bounds. The first panel shows the numerical value of
		\(\Delta e_C\). The second panel shows
		\(\Delta e_C-\Delta e_C^{\rm low}\), and the third panel shows 
		\(\Delta e_C^{\rm up}-\Delta e_C\). Smaller gap values indicate tighter
		bounds. The final panel shows the averaged global output impurity
		\(1-P(\rho_{AB})\).
	}
	\label{fig:2d_damping_bounds}
\end{figure*}

\begin{figure*}[t]
	\centering
	\begin{subfigure}{1\textwidth}
		\centering
		\includegraphics[width=\linewidth]{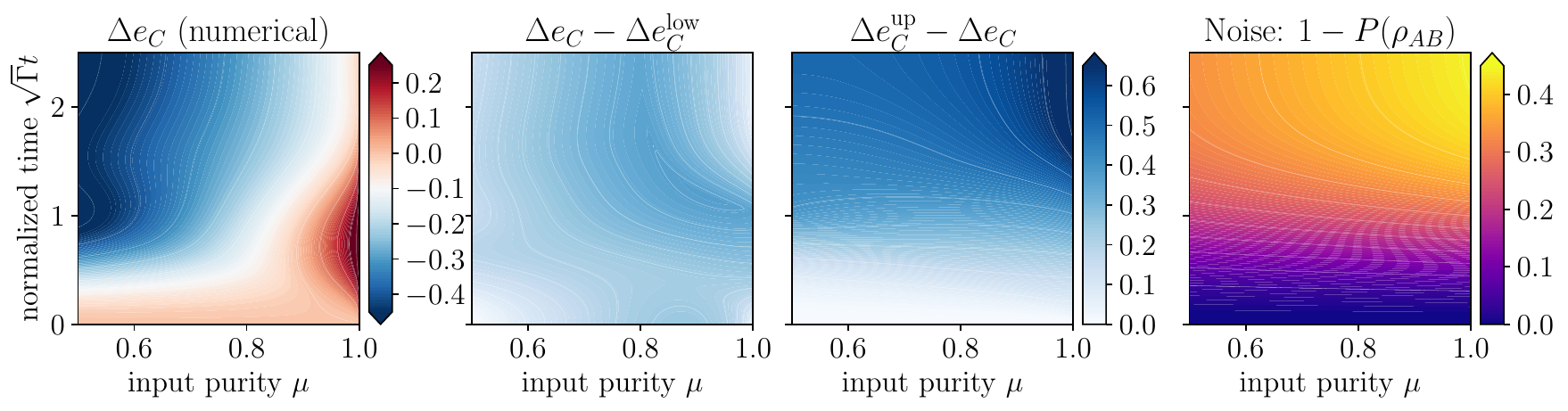}
	\end{subfigure}
	\caption{
		Average concurrence variation for the controlled-phase evolution under correlated phase damping [see Eqs.~\eqref{kraus_damping_cz} and \eqref{kraus_damping_cz2}], shown together with the tightness gaps of the
		analytic bounds. The first panel shows the numerical value of
		\(\Delta e_C\). The second panel shows
		\(\Delta e_C-\Delta e_C^{\rm low}\), and the third panel shows 
		\(\Delta e_C^{\rm up}-\Delta e_C\). Smaller gap values indicate tighter
		bounds. The final panel shows the averaged global output impurity
		\(1-P(\rho_{AB})\).}
	\label{fig:2d_cz_damping_bounds}
\end{figure*}

Following standard methods for Haar integration \cite{Zanardi2000, CollinsSniady2006}, the local unitary twirl over two copies of a bipartite state with reduced purity $\mu$ evaluates to
\begin{widetext}
	\begin{align}
			\Omega(\mu)
		&:=
		\int dU_A dU_B\,
		\left[
		(U_A\otimes U_B)
		\ketbra{\phi_{\vec\lambda}}{\phi_{\vec\lambda}}
		(U_A^\dagger\otimes U_B^\dagger)
		\right]^{\otimes 2}
		\nonumber\\ &=\frac{1}{d^2(d^2-1)^2}
		\left[
		\left(d^2+1-2\mu d\right)
		\left(\mathbbm 1+S_{AA'}S_{BB'}\right)
		+
		\left((d^2+1)\mu-2d\right)
		\left(S_{AA'}+S_{BB'}\right)
		\right].
		\label{eq:Omega-mu}
	\end{align}
\end{widetext}
Here \(S_{AA'}\) and \(S_{BB'}\) swap the corresponding subsystem with its copy. For details, see Appendix~\ref{app:delta-eL:}.

\begin{proposition}
	Let
	\[
	\Phi(\rho)=\sum_\alpha K_\alpha\rho K_\alpha^\dagger
	\]
	be a CPTP map on \(\mathbb C^d\otimes\mathbb C^d\). The orbit-averaged
	linear entropy of the output is
	\begin{align}
		\overline{E_L}_{\rm out}(\Phi,\vec\lambda)
		=
		1-
		\sum_{\alpha,\beta}
		\Tr\!\big[
		(K_\alpha\otimes K_\beta)
		\Omega(\mu)
		\nonumber\\ \times
		(K_\alpha^\dagger\otimes K_\beta^\dagger)
		S_{AA'}
		\big].
		\label{eq:EL-out-Omega}
	\end{align}
	Consequently,
	\begin{align}
		\Delta e_L(\Phi,\vec\lambda)
		=
		\mu
		-
		\sum_{\alpha,\beta}
		\Tr\!\big[
		(K_\alpha\otimes K_\beta)
		\Omega(\mu)
		\nonumber\\ \times
		(K_\alpha^\dagger\otimes K_\beta^\dagger)
		S_{AA'}
		\big].
		\label{eq:delta-eL-Omega}
	\end{align}
    \label{prop:delta-eL}
\end{proposition}

The proof follows from the \textsc{swap} identity
\[
\Tr(\rho_A^2)
=
\Tr\left[
(\rho_{AB}\otimes\rho_{AB})S_{AA'}
\right]
\]
and the second-moment formula [Eq.~\eqref{eq:Omega-mu}]. Details are given in
Appendix~\ref{app:delta-eL:}.

We now derive bounds for the concurrence variation. For an output state
\[
\rho_{AB}
=
\Phi(\ketbra{\psi_\theta(U_A,U_B)}{\psi_\theta(U_A,U_B)}),
\]
define
\begin{equation}
	\delta_P(\rho)
	:=
	\Tr(\rho_{AB}^2)-\Tr(\rho_A^2),
	\quad
	\rho_A=\Tr_B\rho_{AB}.
\end{equation}
As in the product-input case, the lower bound may be strengthened by taking the
maximum of the two versions obtained from tracing over \(A\) or \(B\).
The orbit average of \(\delta_P(\rho)\) reads
\begin{align}
	\overline{\delta_P}(\Phi,\theta)
	:=&
	\int dU_A dU_B\,
	\nonumber\\ &\times
	\delta_P\!\left(
	\Phi(\ketbra{\psi_\theta(U_A,U_B)}{\psi_\theta(U_A,U_B)})
	\right).
	\label{eq:def-average-deltaP-theta}
\end{align}

\begin{proposition}
	\label{prop:average-deltaP-orbit}
	Let \(\Phi\) be a CPTP map with Kraus operators \(\{K_\alpha\}_\alpha\).
	Then
	\begin{align}
		\overline{\delta_P}(\Phi,\vec\lambda)
		=
		\sum_{\alpha,\beta}
		\Tr\!\big[
		(K_\alpha\otimes K_\beta)
		\Omega(\mu)
		(K_\alpha^\dagger\otimes K_\beta^\dagger)
		\nonumber\\ \times
		(S_{AA'}S_{BB'}-S_{AA'})
		\big],
		\label{eq:av-delta-p}
	\end{align}
	where the average is taken over the fixed-Schmidt-spectrum local-unitary orbit
	of \(\ket{\phi_{\vec\lambda}}\).
\end{proposition}

The proof is the same two-copy argument as in
Proposition~\ref{prop:delta-P-tensor}, with \(\Omega_{\otimes}\) replaced by
the orbit second moment \(\Omega(\mu)\). 

Using inequality \eqref{eq:state-tangle-bounds} and averaging over the orbit, we obtain
\begin{equation}
	2\max\{0,\overline{\delta_P}(\Phi,\theta)\}
	\le
	\overline{C^2}_{\rm out}(\Phi,\theta)
	\le
	2\overline{E_L}_{\rm out}(\Phi,\theta).
	\label{eq:tau-out-bounds-theta}
\end{equation}
Since \(0\le C\le1\), one has \(C^2\le C\), and Jensen's inequality gives
\[
\overline{C}_{\rm out}(\Phi,\theta)^2
\le
\overline{C^2}_{\rm out}(\Phi,\theta).
\]
Combining these observations yields the direct concurrence bounds
\begin{equation}
	2\max\{0,\overline{\delta_P}(\Phi,\theta)\}
	\le
	\overline{C}_{\rm out}(\Phi,\theta)
	\le
	\sqrt{
		2\overline{E_L}_{\rm out}(\Phi,\theta)
	}.
	\label{eq:Cout-bounds-theta}
\end{equation}
Equivalently, the average concurrence variation is bounded by 
\begin{widetext}
	\begin{equation}
		2\max\{0,\overline{\delta_P}(\Phi,\theta)\}
		-
		\sin 2\theta
		\le
		\Delta e_C(\Phi,\theta)
		\le
		\sqrt{
			2\left[
			\Delta e_L(\Phi,\theta)
			+
			\frac12\sin^2 2\theta
			\right]
		}
		-
		\sin 2\theta .
		\label{eq:Delta-eC-bounds-theta}
	\end{equation}
\end{widetext}
Equation~\eqref{eq:Delta-eC-bounds-theta} is the main result of this section.
It gives analytic bounds on the average change of concurrence over a
fixed Schmidt orbit for the two-qubit channel \(\Phi\), computable from the Kraus operators. 

For negativity, the present purity-difference method gives no analogous
nontrivial lower bound. However, using \(2\mathcal N(\rho)\le C(\rho)\), one
gets the upper estimate
\begin{equation}
	\Delta e_{\mathcal N}(\Phi,\theta)
	\le
	\frac12
	\sqrt{
		2\left[
		\Delta e_L(\Phi,\theta)
		+
		\frac12\sin^2 2\theta
		\right]
	}
	-
	\frac12\sin 2\theta .
	\label{eq:Delta-eN-upper-theta}
\end{equation}

In Fig.~\ref{fig:2d_damping_bounds} we present a comparison of $\Delta e_C$ with its bounds, in the case of the damping channel. Similarly, Fig.~\ref{fig:2d_cz_damping_bounds} shows the case of the \textsc{cz} channel under phase-damping noise.
The relative tightness of the two bounds depends on the input orbit and the noise strength. The lower bound can be substantially tighter in the noisy regime.

\section{Conclusions}
In this work we investigated two \rd{distinct} aspects of nonunitary bipartite quantum channels: (1) how much entanglement they generate from pure product inputs, and (2) how well they preserve a generic input state from a local unitary orbit, specified by the vector of Schmidt coefficients. 
For arbitrary bipartite channels we derived closed Kraus-operator formulas for the average input-output fidelity $\Favg$ and for its restriction to product inputs, $\Fprod$. \rd{The latter is equivalent to the two-factor specialization of the general multipartite relation between average fidelity and channel fidelity \cite{Demianowicz2010}.} We then showed that the average fidelity for a given local orbit is determined by these two quantities. 
Equivalently, after fixing \(\Fprod\), the purity dependence of the
orbit-averaged fidelity profile is controlled by the single diagnostic
parameter \(\chi_F=\Favg-\Fprod\). Its sign
determines whether the orbit-averaged fidelity is correlated or anticorrelated with the input linear entropy. This fidelity bias characterizes input-orbit dependence and is not an entanglement witness.

We also introduced concurrence- and negativity-based entangling powers for two-qubit channels and compared them with the previously proposed linear-entropy-based quantity. In contrast to the latter, the concurrence- and negativity-based definitions vanish for all separable channels, are convex under channel mixing, and are monotone under local postprocessing. They therefore capture genuine entanglement generation from product inputs, rather than local impurity of mixed outputs, which is the case for the quantity introduced in Ref.~\cite{Kong2015}. This important difference is already apparent in the simple examples considered. The linear-entropy quantity can be nonzero even if no entanglement is generated from product states, while the fidelity-bias parameter 
\(\chi_F\) can still reveal a nontrivial preference for lower-purity or
productlike input orbits.
We note that independently acting local noise typically gives
$\chi_F\leq 0$ in the physically relevant regime. Thus,
$\chi_F>0$ can rule out an independent product-noise description,
but it does not exclude more general separable channels with
product Kraus operators.

    \rd{Zero product-input entangling power does not imply that the channel is separable, as shown by \textsc{swap}. We also established the Choi-state bound $e_{\mathcal N}(\Phi)\leq\mathcal N_{AA':BB'}(\omega_\Phi)$, which makes explicit the distinction between channel nonseparability and entanglement generation by the channel.}

Finally, we derived analytical bounds on concurrence-based entangling power in terms of averaged purity differences and extended the framework to average entanglement variation over general Schmidt orbits. 
In the examples considered, the lower bound provides a useful certificate and can give a tighter estimate than the upper bound in appropriate parameter regimes. 
This provides a natural way to treat channels acting on initially entangled inputs, and shows that the same distinction between fidelity diagnostics and genuine entanglement measures persists beyond the product-input setting. Taken together, our results suggest that a faithful characterization of noisy bipartite operations should combine standard fidelity-based figures of merit with proper entanglement-based diagnostics, since neither class of quantities alone captures the full structure of bipartite quantum noise.

\textit{Acknowledgments--} 
It is a pleasure to thank Rafa{\l} Bistro{\'n} and Jakub Czartowski for useful comments that helped improve this work. We are grateful to Maciej Demianowicz for his remarks.
We acknowledge funding by the European Union under ERC
Advanced Grant TAtypic, Project No. 101142236.
G.T. would like to thank the Jagiellonian University for its warm hospitality and the University of Salerno for managing the period abroad. This research stay was supported by the PNRR (Mission 4: Education and Research).

\onecolumngrid
\appendix
\section{Derivation of fidelity formulas}\label{app:fidelity-derivations}

\subsection{Proof of Proposition~\ref{prop:Favg-Fprod}}
\label{app:proof-prop-Favg-Fprod-general}

We prove separately the formulas for \(\Favg\) and \(\Fprod\).

\paragraph*{1. Average input-output fidelity.}
For a fixed Kraus operator \(K\), define
\begin{equation*}
f(K):=\int d\psi\; \big|\bra{\psi}K\ket{\psi}\big|^2,
\qquad
\Favg(\Phi)=\sum_\alpha f(K_\alpha).
\end{equation*}
Choose an orthonormal basis \(\{\ket{\mu}\}_{\mu=1}^{D}\) of \(\mathcal H_A\otimes\mathcal H_B\), and write
\[
K_{\mu\nu}:=\bra{\mu}K\ket{\nu}.
\]
A Haar-random pure state on \(\mathbb C^{D}\) can be written as
\[
\ket{\psi}=W\ket{0},
\qquad
W\in U(D),
\]
with \(\ket{0}\) a fixed basis vector. Then
\begin{align}
f(K)
=
\int dW\;
\big|\bra{0}W^\dagger K W\ket{0}\big|^2=
\sum_{\mu,\nu,\kappa,\lambda=1}^{D}
K_{\mu\nu}K^*_{\kappa\lambda}
\int dW\;
W^*_{\mu 0}W_{\nu 0}W_{\kappa 0}W^*_{\lambda 0}.
\label{eq:app-global-fK-start-general}
\end{align}
We now use the second-order Weingarten formula \cite{CollinsSniady2006} on \(U(D)\), specialized to the case where all column indices are equal:
\begin{equation}
\int dW\;
W^*_{\mu 0}W_{\nu 0}W_{\kappa 0}W^*_{\lambda 0}
=
\frac{\delta_{\mu\nu}\delta_{\kappa\lambda}+\delta_{\mu\kappa}\delta_{\nu\lambda}}{D(D+1)}.
\label{eq:app-weingarten-one-column-D-general}
\end{equation}
Substituting Eq.~\eqref{eq:app-weingarten-one-column-D-general} into
Eq.~\eqref{eq:app-global-fK-start-general} gives
\begin{align}
f(K)
=
\frac{1}{D(D+1)}
\sum_{\mu,\nu,\kappa,\lambda}
K_{\mu\nu}K^*_{\kappa\lambda}
\left(
\delta_{\mu\nu}\delta_{\kappa\lambda}
+
\delta_{\mu\kappa}\delta_{\nu\lambda}
\right)
=
\frac{|\Tr K|^2+\Tr(KK^\dagger)}{D(D+1)}.
\label{eq:app-global-fK-final-general}
\end{align}
Therefore, we arrive at
\begin{align}
\Favg(\Phi)=
\sum_\alpha f(K_\alpha)=
\frac{1}{D(D+1)}
\left(
\sum_\alpha |\Tr K_\alpha|^2
+
\sum_\alpha \Tr(K_\alpha K_\alpha^\dagger)
\right)=
\frac{D+T}{D(D+1)},
\quad
T:=\sum_\alpha |\Tr K_\alpha|^2.
\end{align}

\paragraph*{2. Product-state average fidelity.}
For a fixed Kraus operator \(K\), define
\begin{equation}
f_\otimes(K)
:=
\int d\psi_A\,d\psi_B\;
\big|\bra{\psi_A\psi_B}K\ket{\psi_A\psi_B}\big|^2,
\qquad
\Fprod(\Phi)=\sum_\alpha f_\otimes(K_\alpha).
\end{equation}
Choose orthonormal bases \(\{\ket{a}\}_{a=1}^{d_A}\) of \(\mathcal H_A\) and
\(\{\ket{b}\}_{b=1}^{d_B}\) of \(\mathcal H_B\), and write
\[
K_{ab,cd}:=\bra{ab}K\ket{cd}.
\]
Let us use the following notation:
\[
\ket{\psi_A}=U\ket{0_A},
\qquad
\ket{\psi_B}=V\ket{0_B},
\qquad
U\in U(d_A),\quad V\in U(d_B).
\]
Then
\begin{equation*}
\bra{\psi_A\psi_B}K\ket{\psi_A\psi_B}
=
\sum_{a,c=1}^{d_A}\sum_{b,d=1}^{d_B}
U^*_{a0}V^*_{b0}\,K_{ab,cd}\,U_{c0}V_{d0}.
\end{equation*}
Hence, we obtain
\begin{align}
f_\otimes(K)
&=
\sum_{a,c,e,g=1}^{d_A}\sum_{b,d,f,h=1}^{d_B}
K_{ab,cd}K^*_{ef,gh}
\int dU\; U^*_{a0}U_{c0}U_{e0}U^*_{g0}
\int dV\; V^*_{b0}V_{d0}V_{f0}V^*_{h0}.
\label{eq:app-fprod-start-general}
\end{align}
Using Eq.~(\ref{eq:app-weingarten-one-column-D-general}) with appropriate dimensions, and performing basic summation over indices, we get an expression,
\begin{equation}
f_\otimes(K)
=
\frac{
|\Tr K|^2
+
\Tr(K^A K^{A\dagger})
+
\Tr(K^B K^{B\dagger})
+
\Tr(KK^\dagger)
}
{d_A d_B (d_A+1)(d_B+1)}.
\label{eq:app-fprod-single-final-general}
\end{equation}
Summing over Kraus operators gives
\begin{align}
\Fprod(\Phi)
=
\sum_\alpha f_\otimes(K_\alpha)
&=
\frac{1}{d_A d_B (d_A+1)(d_B+1)}
\Bigg[
\sum_\alpha |\Tr K_\alpha|^2
+
\sum_\alpha \Tr(K_\alpha^A K_\alpha^{A\dagger})
\nonumber\\
&
+
\sum_\alpha \Tr(K_\alpha^B K_\alpha^{B\dagger})
+
\sum_\alpha \Tr(K_\alpha K_\alpha^\dagger)
\Bigg]\\&=\frac{D+T+M_A+M_B}{D(d_A+1)(d_B+1)},
\end{align}
where
\begin{equation*}
M_A:=\sum_\alpha \Tr(K_\alpha^A K_\alpha^{A\dagger}),
\qquad
M_B:=\sum_\alpha \Tr(K_\alpha^B K_\alpha^{B\dagger}).
\end{equation*}
This finishes the proof of Proposition~\ref{prop:Favg-Fprod}.

\subsection{Fidelity averaged over a Schmidt orbit}
\label{app:fidelity-schmidt-qudit}

In this Appendix we derive Proposition~\ref{prop:Flambda-Kraus} by explicit second-moment Weingarten calculus \cite{CollinsSniady2006}.
Let us recall the basic notation used:
\begin{equation}
\ket{\phi_{\vec\lambda}}
=
\sum_{i=1}^d \sqrt{\lambda_i}\,\ket{ii},
\qquad
\lambda_i\ge 0,
\qquad
\sum_{i=1}^d \lambda_i=1,\qquad\mu=\mu(\vec\lambda)=\sum_{i=1}^d \lambda_i^2.
\end{equation}
For a single Kraus operator \(K\), consider
\begin{align}
f_{\vec\lambda}(K)
:=
\int dU\,dV\;
\left|
\bra{\phi_{\vec\lambda}}
(U^\dagger\otimes V^\dagger)\,K\,(U\otimes V)
\ket{\phi_{\vec\lambda}}
\right|^2,
\label{eq:app-def-fK-lambda}
\end{align}
and denote
\begin{equation}
\Favg_{\vec\lambda}(\Phi)=\sum_\alpha f_{\vec\lambda}(K_\alpha).
\end{equation}
Writing
\begin{equation}
\bra{ii}(U^\dagger\otimes V^\dagger)K(U\otimes V)\ket{jj}
=
\sum_{a,b,c,d}
U^*_{ai}V^*_{bi}\,K_{ab,cd}\,U_{cj}V_{dj},
\end{equation}
we obtain
\begin{align}
f_{\vec\lambda}(K)
&=
\sum_{i,j,k,l}
\sqrt{\lambda_i\lambda_j\lambda_k\lambda_l}
\sum_{\substack{a,b,c,d\\ e,f,g,h}}
K_{ab,cd}\,K^*_{gh,ef}
\int dU\; U^*_{ai}U_{cj}U^*_{ek}U_{gl}
\int dV\; V^*_{bi}V_{dj}V^*_{fk}V_{hl}.
\label{eq:app-fK-lambda-expanded}
\end{align}
For \(U\in U(d)\), the second-moment Weingarten formula \cite{CollinsSniady2006} reads
\begin{align}
\int dU\; U^*_{ai}U_{cj}U^*_{ek}U_{gl}
=
\frac{
\delta_{ac}\delta_{eg}\delta_{ij}\delta_{kl}
+
\delta_{ag}\delta_{ec}\delta_{il}\delta_{kj}
}{d^2-1}
-
\frac{
\delta_{ac}\delta_{eg}\delta_{il}\delta_{kj}
+
\delta_{ag}\delta_{ec}\delta_{ij}\delta_{kl}
}{d(d^2-1)}.
\label{eq:app-U-int-expanded}
\end{align}
The next step is direct summation. Using the following identities,
\begin{align}
\sum_{i,j,k,l}
\sqrt{\lambda_i\lambda_j\lambda_k\lambda_l}\,
\delta_{ij}\delta_{kl}=\left(\sum_i \lambda_i\right)^2=1,
\qquad
\sum_{i,j,k,l}
\sqrt{\lambda_i\lambda_j\lambda_k\lambda_l}\,
\delta_{ij}\delta_{kl}\delta_{il}\delta_{kj}=\sum_i \lambda_i^2=\mu,
\end{align}

\begin{align}
&\sum_{\substack{a,b,c,d\\ e,f,g,h}}
K_{ab,cd}K^*_{gh,ef}\,
\delta_{ac}\delta_{eg}\delta_{bd}\delta_{fh}=
\sum_{a,b,e,f}
K_{ab,ab}K^*_{ef,ef}
=
|\Tr K|^2,\\
&\sum_{\substack{a,b,c,d\\ e,f,g,h}}
K_{ab,cd}K^*_{gh,ef}\,
\delta_{ag}\delta_{ec}\delta_{bh}\delta_{fd}=\sum_{a,b,e,f}
K_{ab,ef}K^*_{ab,ef}=\Tr(KK^\dagger),
\\&\sum_{\substack{a,b,c,d\\ e,f,g,h}}
K_{ab,cd}K^*_{gh,ef}\,
\delta_{ag}\delta_{ec}\delta_{bd}\delta_{fh}=
\sum_{a,e,b,f}
K_{ab,eb}K^*_{af,ef}
=
\Tr(K^A K^{A\dagger}),\\
&\sum_{\substack{a,b,c,d\\ e,f,g,h}}
K_{ab,cd}K^*_{gh,ef}\,
\delta_{ac}\delta_{eg}\delta_{bh}\delta_{fd}=
\sum_{a,e,b,f}
K_{ab,af}K^*_{eb,ef}
=
\Tr(K^B K^{B\dagger}),
\end{align}
where we use the partial trace of bipartite Kraus operators,
\begin{equation}
K^A:=\Tr_B K,
\quad\text{and}\quad
K^B:=\Tr_A K.
\end{equation}
Finally, we obtain
\begin{align}
f_{\vec\lambda}(K)
=
\frac{1}{(d^2-1)^2}
\Bigg[
&\left(1+\frac{1}{d^2}-\frac{2\mu}{d}\right)
\bigl(|\Tr K|^2+\Tr(KK^\dagger)\bigr)
\nonumber\\
&+
\left(\left(1+\frac{1}{d^2}\right)\mu-\frac{2}{d}\right)
\bigl(\Tr(K^A K^{A\dagger})+\Tr(K^B K^{B\dagger})\bigr)
\Bigg].
\label{eq:app-fK-final}
\end{align}

Now let \(\Phi(\rho)=\sum_\alpha K_\alpha \rho K_\alpha^\dagger\).
Summing Eq.~\eqref{eq:app-fK-final} over \(\alpha\) yields
\begin{align}
\Favg_{\vec\lambda}(\Phi)
=
\frac{1}{(d^2-1)^2}
\Bigg[
&\left(1+\frac{1}{d^2}-\frac{2\mu}{d}\right)
\sum_\alpha \bigl(|\Tr K_\alpha|^2+\Tr(K_\alpha K_\alpha^\dagger)\bigr)
\nonumber\\
&+
\left(\left(1+\frac{1}{d^2}\right)\mu-\frac{2}{d}\right)
\sum_\alpha \bigl(\Tr(K_\alpha^A K_\alpha^{A\dagger})+\Tr(K_\alpha^B K_\alpha^{B\dagger})\bigr)
\Bigg].
\end{align}
Making use of
\begin{equation}
T:=\sum_\alpha |\Tr K_\alpha|^2,
\qquad
M_A:=\sum_\alpha \Tr(K_\alpha^A K_\alpha^{A\dagger}),
\qquad
M_B:=\sum_\alpha \Tr(K_\alpha^B K_\alpha^{B\dagger}),
\end{equation}
we arrive at
\begin{equation}
\Favg_{\vec\lambda}(\Phi)
=
\frac{1}{(d^2-1)^2}
\Bigg[
(d^2+T)\left(1+\frac{1}{d^2}-\frac{2\mu}{d}\right)
+
(M_A+M_B)\left(\left(1+\frac{1}{d^2}\right)\mu-\frac{2}{d}\right)
\Bigg],
\end{equation}
which gives the desired Eq.~\eqref{eq:Flambda-Kraus}.

\subsection{Derivation of Proposition~\ref{prop:product-channel-chi}}

For the product channel \eqref{eq:product-channel}, a Kraus representation is
\begin{equation}
K_{ij}=A_i\otimes B_j.
\end{equation}
The global invariant $T$ defined in Proposition~\ref{prop:Favg-Fprod} factorizes:
\begin{align}
T
&=
\sum_{i,j}\left|\Tr(A_i\otimes B_j)\right|^2
=
\sum_{i,j}\left|\Tr A_i\right|^2\left|\Tr B_j\right|^2
=
x_Ax_B.
\label{eq:T-factorized}
\end{align}
Likewise, using $\Tr_B(A_i\otimes B_j)=(\Tr B_j)A_i$ and $\Tr_A(A_i\otimes B_j)=(\Tr A_i)B_j$, one finds
\begin{align}
M_A
&=
\sum_{i,j}\Tr\!\left[(\Tr_B K_{ij})(\Tr_B K_{ij})^\dagger\right]\nonumber
=
\sum_{i,j}|\Tr B_j|^2\,\Tr(A_iA_i^\dagger)
=
x_B\sum_i \Tr(A_iA_i^\dagger)
\\&=
x_B\,\Tr \Bigl(\sum_i A_i^\dagger A_i\Bigr)
=
2x_B,
\quad M_B
=
2x_A.
\label{eq:MB-product}
\end{align}
Substituting Eqs.~\eqref{eq:T-factorized} and \eqref{eq:MB-product} into Proposition~\ref{prop:Favg-Fprod} gives
\begin{align}
\Favg(\Phi)&=\frac{4+x_Ax_B}{20},
\nonumber\\
\Fprod(\Phi)&=\frac{4+x_Ax_B+2x_A+2x_B}{36},
\end{align}
and hence we arrive at Eq.~\eqref{eq:chi-product-x}:
\begin{equation}
\chi_F(\Phi)
=
\frac{8+2x_Ax_B-5x_A-5x_B}{90}.
\end{equation}
To obtain  Eq.~\eqref{eq:chi-product-f} we use $x_A=6f_A-2$ and $x_B=6f_B-2$. Finally, let
\begin{equation}
u:=1-f_A,
\qquad
v:=1-f_B.
\end{equation}
Under assumption \eqref{eq:condition-fA-fB} one has $u,v\in[0,1/2]$, and Eq.~\eqref{eq:chi-product-f} becomes
\begin{equation}
\chi_F(\Phi)=\frac{4uv-u-v}{5}.
\end{equation}
Since $u,v\in[0,1/2]$, one has
\begin{equation}
4uv\le u+v,
\end{equation}
and therefore $\chi_F(\Phi)\le 0$.

\section{Entangling power for controlled gates}
\label{app:1}

For a Haar-random single-qubit state $\ket{\psi}=a\ket0+b\ket1$ one has
$|a|^2\sim\mathrm{Unif}[0,1]$ and therefore
\begin{align}
&\mathbb{E}\big[|a|^2|b|^2\big]
=\int_0^1 x(1-x)\,dx=\frac{1}{6},
\label{eq:avg_a2b2}\\
&\mathbb{E}\big[|a||b|\big]
=\int_0^1 \sqrt{x(1-x)}\,dx=\mathrm{B}\!\left(\frac{3}{2},\frac{3}{2}\right)
=\frac{\pi}{8},
\label{eq:avg_ab}
\end{align}
where $\mathrm{B}$ denotes the beta function of Euler.

Let $\mathrm{\textsc{cp}}(\varphi)=\mathrm{diag}(1,1,1,e^{i\varphi})$ and take a product input
$\ket{\psi_A}=a\ket0+b\ket1$, $\ket{\psi_B}=c\ket0+d\ket1$.
Then
\begin{equation}
\ket{\Psi_{\rm out}}=ac\ket{00}+ad\ket{01}+bc\ket{10}+e^{i\varphi}bd\ket{11}.
\end{equation}
Writing the $2\times 2$ coefficient matrix
\[
A=\begin{pmatrix}
ac & ad\\
bc & e^{i\varphi}bd
\end{pmatrix},
\]
the concurrence of a two-qubit pure state reads $C=2|\det A|$; hence,
\begin{equation}
C=2|abcd|\ |e^{i\varphi}-1|=4|abcd|\ |\sin(\varphi/2)|.
\label{eq:CP_C_pointwise_app}
\end{equation}
Because the input Haar measure factorizes between $A$ and $B$,
Eq.~\eqref{eq:CP_C_pointwise_app} gives
\begin{align}
e_C(\Ad_{\mathrm{\textsc{cp}}(\varphi)})
&=\int d\psi_A d\psi_B\ C
=4|\sin(\varphi/2)|\Big(\mathbb{E}[|a||b|]\Big)\Big(\mathbb{E}[|c||d|]\Big)
\nonumber\\
&=4|\sin(\varphi/2)|\left(\frac{\pi}{8}\right)^2
=\frac{\pi^2}{16}\,|\sin(\varphi/2)|.
\label{eq:eC_CP_app}
\end{align}
Since output states are pure, negativity $\mathcal N=C/2$ pointwise; hence,
\begin{equation}
e_{\mathcal N}(\Ad_{\mathrm{\textsc{cp}}(\varphi)})=\frac{1}{2}e_C(\Ad_{\mathrm{CP}(\varphi)})
=\frac{\pi^2}{32}\,|\sin(\varphi/2)|.
\label{eq:eN_CP_app}
\end{equation}
Using $E_L=C^2/2$ for pure outputs and Eq.~\eqref{eq:CP_C_pointwise_app},
\begin{align}
e_L(\Ad_{\mathrm{\textsc{cp}}(\varphi)})
&=\int d\psi_A d\psi_B\ \frac{1}{2}C^2
=\frac{1}{2}\int d\psi_A d\psi_B\ 16|abcd|^2\sin^2(\varphi/2)
\nonumber\\
&=8\sin^2(\varphi/2)\Big(\mathbb{E}[|a|^2|b|^2]\Big)^2
=8\sin^2(\varphi/2)\left(\frac{1}{6}\right)^2
=\frac{2}{9}\sin^2(\varphi/2).
\label{eq:eL_CP_app}
\end{align}
\textsc{cz} is $\mathrm{\textsc{cp}}(\pi)$, so $|\sin(\pi/2)|=1$. Equations
\eqref{eq:eC_CP_app} and \eqref{eq:eN_CP_app} give
$e_C(\Ad_{\mathrm{\textsc{cz}}})=\pi^2/16$ and $e_{\mathcal N}(\Ad_{\mathrm{\textsc{cz}}})=\pi^2/32$.

\section{Concurrence-based entangling power for controlled-Z gate under correlated dephasing}\label{app:concurrence-rescaling}
For a product input \(\ket{\psi_A}\otimes\ket{\psi_B}\) with
\[
\ket{\psi_A}=a\ket{0}+b\ket{1},
\qquad
\ket{\psi_B}=c\ket{0}+d\ket{1},
\]
the ideal \textsc{cz} output reads
\begin{equation}
\ket{\Psi_+}
=
\mathrm{\textsc{cz}}\ket{\psi_A\psi_B}
=
ac\ket{00}+ad\ket{01}+bc\ket{10}-bd\ket{11}.
\label{eq:Psi-plus-CZ}
\end{equation}
Let
\begin{equation}
\ket{\Psi_-}:=(Z\otimes Z)\ket{\Psi_+}
=
ac\ket{00}-ad\ket{01}-bc\ket{10}-bd\ket{11}.
\label{eq:Psi-minus-CZ}
\end{equation}
Since the correlated dephasing channel can be written as
\begin{equation}
\Phi_u(\rho)=\frac{1+u}{2}\,\rho+\frac{1-u}{2}\,(Z\otimes Z)\rho(Z\otimes Z),
\label{eq:Phi-u-mixture-app}
\end{equation}
the noisy \textsc{cz} output is
\begin{equation}
\rho_u
=
\Phi_u^{\mathrm{\textsc{cz}}}(\ketbra{\psi_A\psi_B}{\psi_A\psi_B})
=
p\ketbra{\Psi_+}{\Psi_+}+(1-p)\ketbra{\Psi_-}{\Psi_-},
\qquad
p=\frac{1+u}{2}.
\label{eq:rho-u-decomp}
\end{equation}
Now we use the Wootters decomposition formula \cite{Wootters1998}. For a decomposition
\(\rho=\sum_i p_i \ketbra{x_i}{x_i}\), define
\begin{equation}
\tau_{ij}:=\sqrt{p_i p_j}\,\braket{x_i}{\widetilde{x_j}},
\qquad
\ket{\widetilde{x}}:=(\sigma_y\otimes\sigma_y)\ket{x^*}.
\end{equation}
The concurrence is the difference of the singular values of \(\tau\)
for a rank-2 decomposition \cite{Wootters1998}.
For the states \(\ket{\Psi_\pm}\) above one finds
\begin{align}
\braket{\Psi_+}{\widetilde{\Psi_+}}
&=
4(a b c d)^*,
\\
\braket{\Psi_-}{\widetilde{\Psi_-}}
&=
4(a b c d)^*,
\\
\braket{\Psi_+}{\widetilde{\Psi_-}}
&=
\braket{\Psi_-}{\widetilde{\Psi_+}}
=
0.
\end{align}
Therefore, with
\[
C_0:=C(\ket{\Psi_+})=C(\ket{\Psi_-})=4|abcd|,
\]
the matrix \(\tau\) is diagonal, and up to irrelevant phase factor, could be written as
\begin{equation}
\tau=
\begin{pmatrix}
p\,C_0 & 0\\
0 & (1-p)\,C_0
\end{pmatrix},
\end{equation}
and Wootters's formula gives
\begin{equation}
C(\rho_u)=\bigl|p-(1-p)\bigr|\,C_0=u\,C_0.
\end{equation}
Hence, correlated dephasing rescales the output concurrence pointwise by a coefficient \(u\):
\begin{equation}
C\!\left(\Phi_u^{\mathrm{\textsc{cz}}}(\ketbra{\psi_A\psi_B}{\psi_A\psi_B})\right)
=
u\,
C\!\left(\Ad_{\mathrm{\textsc{cz}}}(\ketbra{\psi_A\psi_B}{\psi_A\psi_B})\right).
\label{eq:concurrence-rescaling-u}
\end{equation}
Averaging over product-input states yields
\[
e_C(\Phi_u^{\mathrm{\textsc{cz}}})=u\,e_C(\Ad_{\mathrm{\textsc{cz}}}).
\]

\section{Proof of Proposition~\ref{prop:delta-eL}\label{app:delta-eL:}}

Once Eq.~(\ref{eq:Omega-mu}) is proved, the proof of Proposition~\ref{prop:delta-eL} follows by using standard techniques for Haar integration. 
Consider the generic state on the Schmidt orbit,
\begin{equation}
    \ket{\psi_{AB}}=(V_A\otimes V_B)\sum_{i=0}^{d-1}\sqrt{\lambda_i}\ket{i}\ket{i}.
\end{equation}
The quantity to evaluate is 
\begin{equation}
    \Omega(\mu)= \int dV_A dV_B \: \ketbra{\psi_{AB}}{\psi_{AB}}^{\otimes 2}
    =\int dV_A dV_B \left[\sum_{ij}\sqrt{\lambda_i \lambda_j} (V_A \ketbra{i}{j}V_A^\dagger ) \otimes ( V_B \ketbra{i}{j} V_B^\dagger)\right]^{\otimes2},
\end{equation}
which can also be rewritten by factorizing the $AA'$ contribution from the $BB'$ contribution,
\begin{equation}
    \Omega(\mu)= \sum_{ijkl}\sqrt{\lambda_i \lambda_j \lambda_k \lambda_l}
    \left[\int d V_A V_A \ketbra{i}{j}V_A^\dagger  \otimes V_A \ketbra{k}{l}V_A^\dagger \right]^{\otimes 2}.
\end{equation}
Using Weingarten calculus \cite{CollinsSniady2006}, one obtains 
\begin{equation}
\begin{split}
  \left[\int d V_A V_A \ketbra{i}{j}V_A^\dagger  \otimes V_A \ketbra{k}{l}V_A^\dagger \right]
  &=\frac{(d\Tr\ketbra{i}{j} \Tr\ketbra{k}{l} - \Tr \ketbra{i}{j}\ketbra{k}{l} )\mathbbm 1_{AA'} + (d\Tr \ketbra{i}{j}\ketbra{k}{l} - \Tr \ketbra{i}{j} \Tr \ketbra{k}{l})S_{AA'}}{d(d^2-1)}\\
  &=\frac{(d\delta_{ij}\delta_{kl} - \delta_{il}\delta_{jk})\mathbbm 1_{AA'} + (d \delta_{il}\delta_{jk}-\delta_{ij}\delta_{kl})S_{AA'}}{d(d^2-1)}.
  \end{split}
\end{equation}
Taking the tensor product of the corresponding \(A\) and \(B\) averages gives
\begin{equation}
    \begin{split}
    &\frac{(d^2\delta_{ij}\delta_{kl} + \delta_{il}\delta_{jk}-2d\delta_{ij}\delta_{kl}\delta_{ik})\mathbbm 1_{AA'BB'} 
    + (d^2 \delta_{il}\delta_{jk}+\delta_{ij}\delta_{kl} - 2d\delta_{il} \delta_{jk}\delta_{ik})S_{AA'BB'}
     }{d^2(d^2-1)^2} 
     \\
     &+\frac{(d^2 \delta_{ij}\delta_{kl}\delta_{ik}+\delta_{il}\delta_{jk}\delta_{ik}-d(\delta_{ij}\delta_{kl}+\delta_{il}\delta_{jk}))\mathbbm 1_{AA'}S_{BB'}+(d^2 \delta_{ij}\delta_{kl}\delta_{ik}+\delta_{il}\delta_{jk}\delta_{ik}-d(\delta_{ij}\delta_{kl}+\delta_{il}\delta_{jk}))S_{AA'}\mathbbm 1_{BB'}}{d^2(d^2-1)^2}.
     \end{split}
\end{equation}
The contribution from $\mathbbm1_{AA'BB'}$ and $S_{AA'BB'}$ gives
\begin{equation}
   x= (d^2+1) \underbrace{\sum_{ij}\lambda_i \lambda_j}_{=1} - 2d \underbrace{\sum_{i}\lambda_i^2}_\mu,
\end{equation}
while the contribution from $\mathbbm 1_{AA'}S_{BB'}$ and $S_{AA'}\mathbbm 1_{BB'}$ is 
\begin{equation}
   y=  (d^2+1)\underbrace{\sum_i \lambda_i ^2}_\mu
   -2d \underbrace{\sum_{ij}\lambda_i \lambda_j}_{=1}.
\end{equation}
In the case of separable states we have $x=y=(d^2+1-2d)=(d-1)^2$
and we obtain $\Omega=(\mathbbm 1_{AA'} + S_{AA'})(\mathbbm1_{BB'}+S_{BB'})/(d^2(d+1)^2)$, agreeing with Ref.~\cite{Zanardi2000}.
The final expression is
\begin{equation}
    \Omega(\mu)=\frac{1}{d^2(d^2-1)^2}\left[\left(d^2+1-2\mu d\right)(\mathbbm 1_{AA'BB'} + S_{AA'BB'})+ \left(\mu d^2 + \mu - 2d\right)\left(\mathbbm 1_{AA'}\otimes S_{BB'}+S_{AA'}\otimes \mathbbm 1_{BB'}\right)\right].
\end{equation}

\bibliography{ref.bib}
\end{document}